%% file: 00-main-IncrementalSafe.tex
\documentclass[copyright,creativecommons]{eptcs}
\providecommand{\event}{GCM 2020} 

\usepackage{etex}   
\usepackage{amsthm}\input{amsthmdefs}
\usepackage{hyperref}

\usepackage[all]{xy}
\usepackage{tikz-uml}

\let\umlnote\relax
\usepackage{tikzmacros} 
\usepackage{longtable}
\usepackage{multirow}

\usepackage{amssymb,amsmath}
\usepackage[T1]{fontenc}
\usepackage{bm}
\newcommand{\spv}{\vspace{.1cm}}

\newcommand{\nacs}{{\mathcal N}}
\newcommand{\Rrel}[1]   {\stackrel{{#1}}{\Longrightarrow}}
\newcommand{\rrel}[1]   {\stackrel{{#1}}{\rightarrow}}
\newcommand{\lrel}[1]   {\stackrel{{#1}}{\leftarrow}}

\usepackage{graphicx}
\newcounter{comCounter}[page]
\newcommand{\ifnonempty}[2]{\ifthenelse{\equal{#1}{}}{}{#2}}

\newcommand{\AC}[1]{{#1}}
\newcommand{\from}{\leftarrow}
\newdir{ >}{*{}!/-.2em/\dir{>}}
\newcommand{\ov}{\overline}
\newcommand{\cycl}[1]{\ensuremath{\mbox{\textcircled{\scriptsize{$#1$}}}}}

\newcommand{\inv}[1]{{\mathit{inv}(#1)}}

\title{Encoding Incremental NACs in Safe Graph Grammars\\ using Complementation}
\author{Andrea Corradini
\institute{Dipartimento di Informatica,  University of Pisa, Italy}
\email{andrea@di.unipi.it}
\and
Maryam Ghaffari Saadat \qquad\qquad Reiko Heckel
\institute{Department of Informatics, University of Leicester, UK}
\email{\{mgs17,rh122\}@leicester.ac.uk}}
\def\titlerunning{Encoding Incremental NACs}
\def\authorrunning{A. Corradini, M. Ghaffari Saadat  \& R. Heckel}
\begin{document}
\maketitle             
\begin{abstract}
In modelling complex systems with graph grammars (GGs), it is convenient to restrict the application of rules using attribute constraints and negative application conditions (NACs). However, having both attributes and NACs in GGs renders the behavioural analysis (e.g. unfolding) of such systems more complicated. We address this issue by an approach to encode NACs using a complementation technique. We consider the correctness of our encoding under the assumption that the grammar is safe and NACs are incremental, 
and outline how this result can be extended to unsafe, attributed grammars.  
 \end{abstract}

\input{01-Introduction}

\input{02-Background}

\input{03-Encoding}

\input{04-Equivalence}
\input{04-SymbolicAttributed}

\input{05-Conclusion}

%
%
\bibliographystyle{eptcs}
\bibliography{references}
%
%

\end{document}

%% file: amsthmdefs.tex
\newtheorem{theorem}{Theorem}[section]
\newtheorem{lemma}[theorem]{Lemma}
\newtheorem{proposition}[theorem]{Proposition}
\newtheorem{corollary}[theorem]{Corollary}

\theoremstyle{definition}
\newtheorem{definition}[theorem]{Definition}
\newtheorem{example}[theorem]{Example}
\newtheorem{construction}[theorem]{Construction}
\newtheorem{fact}[theorem]{Fact}

\theoremstyle{remark}



%% file: 01-Introduction.tex


\section{Introduction}
\label{sec:Introduction}

\medskip
\noindent
\emph{Motivation.}
Graph grammars (GGs) can model complex systems or networks where data transformations are tightly coupled with structural changes. A defining feature of complex systems is that they show emergent behaviour arising from the actions of decentralised, distributed, often autonomous agents individually described by means of simple rules~\cite{pinaud2012porgy,koch2002,book}. One purpose of modelling such systems is to investigate this emergent behaviour. This requires semantic models and analysis techniques able to account for concurrency properties, such as the unfolding, which captures in a single structure the full behaviour of the system, and its partial approximations~\cite{BaldanCMR07,DBLP:journals/iandc/BaldanCK08}.

It is often natural to impose conditions on the rules governing the actions of individual agents to restrict their application using both constraints on data attributes and negative application conditions (NACs). Although allowing for both attributes and NACs enhances the expressiveness of rules and facilitates a concise style of modelling, it renders semantics and behavioural analysis more complicated. 

In previous work, authors of this paper have generalised the theory of unfolding from plain  graph grammars separately to the case with NACs and to the attributed case~\cite{corradini2019unfolding,attunfolding}. While in the plain case, conflicts and dependencies between transformation steps are based on how nodes and edges are used, created or deleted by these steps, e.g.~if created by one step and used by another, in the case with NACs we have to account for additional sources of conflicts and dependencies, e.g.~when a rule deletes part of a structure that inhibits the application of another due to a NAC. Also, in the attributed case, dependencies can be based on attributes updated by one step and then read by another. Hence both ge\-ne\-ralisa\-tions give rise to additional complexity both conceptually and in terms of the construction of unfolding.

Our ultimate goal is to provide a comprehensive semantics and analysis approach based on a theory of unfolding that supports both attribution and NACs. 
As a first contribution, we investigate here the possibility of encoding incremental NACs of a safe grammar (where, up to renaming, all reachable graphs are subgraphs of the type graph) by exploiting structural properties of the 
graphs.  
We also sketch how this technique could be extended to unsafe grammars by exploiting attributes.

\medskip
\noindent
\emph{Related work.}  In~\cite{da2017theorem} the authors translate graph grammars with NACs and attributes into event-B models, a state-based formalism based on first-order logic with set theory. The authors prove in two stages that the encoded model is semantically equivalent to the input grammar: 1) encoding NACs only, 2) extending the encoding to attributes. Their goal is to be able to use theorem proving instead of model checking to show behavioural correctness without exploring the state space. \AC{Even if our approach is different, our motivations overlap with theirs in attempting to reduce complexity of behavioural analysis.}


Another example for transforming negative application conditions over attributed graphs into a se\-mantics-preserving way is~\cite{DBLP:journals/jlp/NassarKAT20}. Here, the motivation is to simplify rather than encode the conditions, \AC{but the technique of transforming rules in order to ensure that certain constraints are preserved is quite similar to the construction presented later in Section~\ref{sec:SafeCase}.}



\medskip
\noindent
\emph{Methodology.}
Our problem can be expressed generically as follows: Given a \emph{conditional} graph grammar 
$\mathcal{CG}$, that is a graph grammar where rules may be equipped with NACs  constraining their applicability,   
%
construct an \emph{unconditional} grammar $DE(\mathcal{CG})$ 
that is equivalent to $\mathcal{CG}$ in both the derivations generated and their dependency relations. That means, if two consecutive steps in a derivation of $DE(\mathcal{CG})$ are sequentially dependent (cannot be swapped), the same should hold for the corresponding steps in the derivation over $\mathcal{CG}$, and vice versa.  

To express relations between grammars we rely on a notion of grammar morphism that preserves the behaviour of the source grammar, inducing a mapping of derivations and, moreover, preserves the relation of sequential independence between steps within a derivation. Such a morphism is based on a relation between the type graphs of the  grammars that allows, among other things, to drop or rename types, and a mapping of rules that permits to reduce rules by dropping context and NACs.

The transformation of a given conditional grammar $\mathcal {CG}$ to its encoding $DE(\mathcal {CG})$ without NACs consists of two steps. First, we define a grammar $E(\mathcal {CG})$ whose NACs are essentially the same as in $\mathcal {CG}$ \AC{but are also encoded in a suitable way in the graphs to be transformed. Grammar $E(\mathcal {CG})$ is related to the original grammar by a morphism $e: E(\mathcal {CG}) \to \mathcal{CG}$.
It can be shown that the NACs of  $E(\mathcal {CG})$ are redundant thanks to their encoding, thus we can remove them from $E(\mathcal {CG})$ obtaining a grammar $DE(\mathcal {CG})$, and we can show that there is a  morphism $d: E(\mathcal {CG}) \to DE(\mathcal{CG})$ which is the identity on types, constraints, rules and input graph.}
%
%
%
Our goal then is to show that both grammar morphisms of the span
$\mathcal {CG} \lrel{e} E(\mathcal{CG}) \rrel{d} DE(\mathcal{CG})$
not only \emph{preserve} derivations and sequential independence, as every morphism, but also \emph{reflect} them. In particular, this means that 
\begin{enumerate}
    \item 
the NACs dropped by $d$ do not restrict the applicability of rules in $E(\mathcal {CG})$ nor create dependencies and conflicts that do not exist on the underlying transformations;
    \item 
the encoding of $\mathcal {CG}$ in $E(\mathcal{CG})$ via $e$ does not restrict the applicability of rules nor create new dependencies and conflicts.
\end{enumerate}
Since preservation of independence implies  reflection of causality and conflicts, it ensures that $DE(\mathcal {CG})$ and $\mathcal {CG}$ have not only the same sequential behaviour, but also that their computations have the same branching structure, and thus equivalent unfoldings.

Our approach for showing that morphisms $d$ and $e$ reflect derivations relies on showing that (1) the NACs are redundant in $E(\mathcal{CG})$ so dropping them in $DE(\mathcal{CG})$ does not change the behaviour and that (2) \AC{the structure resulting from the encoding of the NACs in} $E(\mathcal{CG})$ does not restrict the behaviour more than it was already restricted by the NACs in $\mathcal{CG}$.  Concerning independence, it is trivially reflected by $d$ but for morphism $e$ the property is at present only conjectured.
%
Both arguments will rely on \AC{suitable properties satisfied by the reachable graphs} in  grammar $E(\mathcal{CG})$. To express such properties formally 
we introduce a set of graph constraints and require that such constraints are invariants, that is, they are satisfied in the start graph and preserved by the application of rules, and hence hold for all reachable graphs. This will play an important role in verifying the correctness of the encoding.

\medskip
\noindent
\emph{Realisation.}
In particular, we propose an approach to encode \emph{incremental} NACs (i.e. NACs that can only be created or destroyed in a certain order~\cite{incrementalnacs}) into additional rule context.
The restriction to \emph{incremental} NACs is motivated by the fact that only with this restriction the notion of sequential independence (and therefore of causality) among steps enjoys some expected properties, which guarantee that a concurrent semantics can be defined properly~\cite{DBLP:conf/gg/CorradiniH14}. Our construction can be considered as a generalisation of the well-known technique of complementation for Elementary Net Systems \cite{DBLP:conf/ac/Rozenberg86}, which transforms a system with possible situations of contact (where we try to place tokens on places that are already marked) into an equivalent one which is contact-free. 

We consider the encoding in detail for safe conditional graph grammars. This is a prerequisite to a semantic comparison between notions of unfolding with~\cite{corradini2019unfolding} and without NACs~\cite{DBLP:conf/tagt/BaldanCM98}, all resulting in safe (i.e., occurrence) graph grammars. After the presentation of the relevant background in Section~\ref{sec:background}, the complementation construction is  described in Section~\ref{sec:encoding-safe}. The main results, showing the equivalence of the original grammar with the one obtained by encoding the NACs and then forgetting them are presented in  Section~\ref{sec:equivalenceOfGrammars}. A simple safe grammar modeling a client-server architecture for online meetings is used to illustrate concepts and constructions. 
Then in Section~\ref{sec:AttributedEncoding} we discuss informally the general case of (unsafe) conditional attributed graph grammars and their encoding. Here we introduce the idea of using attributes as counters for NAC occurrences reminiscent of reference counters in garbage collection algorithms. We also rely on rules with multiobjects~\cite{book} to manage the creation and deletion of complemented edges, but consider such rules as rule schemata that can be instantiated to countable sets of basic rules. This is described with an example of a Token-Curated Registry (TCR), an architectural pattern for smart contracts. 
In Section~\ref{sec:Conclusion} we summarise our contributions and outline future work.

%% file: 02-Background.tex


\section{Background: Conditional Graph Grammars and Their Morphisms}
\label{sec:background}
This section first summarizes the basic definitions of typed graph grammars~\cite{DBLP:journals/fuin/CorradiniMR96} based on the DPO approach~\cite{DBLP:conf/focs/EhrigPS73,DBLP:series/eatcs/EhrigEPT06} extended by negative application conditions (NACs)~\cite{DBLP:journals/fuin/HabelHT96}. In the second part we introduce a quite general definition of conditional grammar morphisms and we prove that they preserve derivations and independence.

\medskip
\noindent
\emph{Graphs and typed graphs.}
Formally, a  \emph{(directed, unlabelled) graph} is a tuple $G =
\langle N, E, s, t \rangle$, where $N$ is a set of  \emph{nodes},
$E$ is a set of  \emph{arcs}, and $s, t: E \rightarrow N$ are the
 \emph{source} and  \emph{target} functions.
A  \emph{(graph) morphism} $f: G \rightarrow G'$ is a pair of functions $f
= \langle f_N: N \rightarrow N', f_E: E \rightarrow E'\rangle$
preserving sources and targets, i.e., such that $f_N \circ s = s' \circ
f_E$ and $f_N \circ t = t' \circ f_E$. Morphism $f$ is a \emph{mono(morphism)} if both components are injective, and is an \emph{iso(morphism)} if both are bijective.

The category of graphs and graph morphisms is denoted by $\mathbf{Graphs}$. Given a graph $TG$, called \emph{type graph}, a  \emph{$TG$-typed (instance) graph} is a pair
$\langle G, t_G \rangle$, where $G$ is a graph and $t_G : G \to TG$ is
a (typing) morphism. A morphism between typed graphs $f : \langle G_1, t_{G_1}
\rangle \to \langle G_2, t_{G_2} \rangle$ is a graph morphisms $f :
G_1 \to G_2$ consistent with the typing, i.e., such that $t_{G_1} =
t_{G_2} \circ f$. The category of $TG$-typed graphs and typed graph
morphisms is denoted by  $\mathbf{Graphs}_{TG}$.

\medskip
\noindent
\emph{Double-pushout rewriting.} A \emph{(graph transformation) rule}  $p = (L \lrel l K \rrel r R)$ is a span of monos in $\mathbf{Graphs}_{TG}$.
A \emph{match} of rule $p$ in a graph $G$ is a mono $m: L \to G$.
Given a match $m$ of $p$ in $G$,  
a \emph{double-pushout (DPO) transformation} $G \Rrel{p,m}  H$ from $G$ to  $H$ exists if we can construct a diagram such as~\eqref{dia:derivation-diagram-def} where both squares are pushouts in  $\mathbf{Graphs}_{TG}$.

\medskip
\noindent
\begin{minipage}[c]{.65\linewidth}
\emph{Derived and ancestor rules.} If two rules $p = (L \lrel l K \rrel r R)$ and $p' = (G \lrel g D \rrel h H)$ are related by a DPO diagram as in \eqref{dia:derivation-diagram-def}, then we say that $p'$ \emph{is derived from} $p$ and $p$ \emph{is an ancestor of} $p'$.
\end{minipage}
\begin{minipage}[c]{.34\linewidth}
    \vspace{-2ex}
    \begin{align}
      \label{dia:derivation-diagram-def}
      \begin{array}[c]{l}
        \xymatrix@R=4ex{
          L 
          \ar[d]_{m}
          & 
          K
          \ar[l]_{l}
          \ar[r]^{r}
          \ar[d]^{d}
          &
          R
          \ar[d]^{m^*}
          \\
          G 
          &
          D
          \ar[l]^{g}
          \ar[r]_{h}
          &
          H
        }
      \end{array}
    \end{align}
  \end{minipage}

\medskip
\noindent
\emph{Negative constraints and incrementality.} The applicability of rules can be restricted by (negative) constraints.
%
%
 A \emph{constraint} over a rule  $p$  is a mono $n: L \to
N$. A match $m: L \to G$ \emph{satisfies} $n$ (written $m \models n$) if there is no mono $q: N \to G$ such that $q \circ
n = m$ (see Diagram~\eqref{dia:nac-ipo} (left)). If $\nacs$ is a set of constraints we write $m \models \nacs$ if $m \models n$ for each $n \in \nacs$.
If $n$ and $n'$ are  constraints over $p$ we say that $n$ \emph{subsumes} $n'$ (written $n \models n'$) if for every match $m: L \to G$, $m \models n$ implies $m \models n'$.
   \begin{align}
      \label{dia:nac-ipo}
      \begin{array}[c]{l@{\hspace{2cm}}r}
\xymatrix{N  \ar@{.>}[dr]_{q}|{/} &  L \ar[l]_{n} \ar[d]^{m}\\   & G} 
&
        \xymatrix@R=2ex{
          N  & & L \ar[ll]_{n}  \\ 
          & A \ar[ul] & & B \ar[ll] \ar[ul]\\
          N^{-} \ar[uu] \ar@{-->}[ur]& & L^{-} \ar[ll]_{n^{-}}  \ar[uu]|(.5)\hole \ar@{-->}[ur] 
          }
       \end{array}
    \end{align}
We will often refer, for a constraint $n: L \to N$,  to its ``negative items'', that is, to the nodes and edges of $N$ that are not in the image of $L$. The smallest subgraph of $N$ containing those items will be denoted $N^{-}$, and is characterized as in Diagram~\eqref{dia:nac-ipo} (right) by an {initial pushout}: The square made of morphisms $n$ and $n^{-}: L^{-} \to N^{-}$ and of the vertical morphisms 
is an \emph{initial pushout} over $n$ if for each pushout square $N \from A \from  B \to  L \rrel n N$ with $B \to L$ mono, there are unique morphisms $N^{-} \to A$ and $L^{-}  \to B$ making the diagram commute.
For a constraint $n$ we will call $n^{-}$ its \emph{shape},  $L^{-}$ its \emph{border} and $N^{-}$ its \emph{body}.
Note that a shape would be an iso iff $n$ is an iso, but in this case the rule is never applicable and can be dropped. Thus we will safely assume that the shape of a constraint is not an iso. For a set of constraints $\mathcal{N}$  we denote by $sh(\mathcal{N})$ its set of shapes. 

A  constraint $n: L \to N$ is \emph{incremental} if whenever it can be decomposed in two different ways, i.e. $L \rrel n N$ = $L \rrel {a_{1}} X_{1} \rrel {b_{1}} N$ = $L \rrel {a_{2}} X_{2} \rrel {b_{2}} N$  
where all morphisms are mono, then there exists either a morphism from $X_{1}$ to $X_{2}$ or one in the opposite direction making the two triangles commute. For untyped graphs it is possible to show that if $n$ is incremental then its shape $n^{-}$ can only be one of those shown in~\eqref{eq:shapes}, where nodes are boxes, $L^{-}$ is made of the black items, while the red ``negative'' items belong to $N^{-}\setminus L^{-}$. For $TG$-typed graphs, for every shape in~\eqref{eq:shapes} each distinct typing morphism $t_{N^{-}}: N^{-} \to TG$ determines a different incremental shape.

\begin{equation}
\label{eq:shapes}
\scalebox{.8}{
\begin{mytikz}
\graphsingle (g1) {
\graphname{shape $\mathit{IN}$:}
    \begin{mytikz}  
	\single (s1) at (0,0) {}; 
    \begin{naconly}
	\single (m1) at (1,0) {};   
    \edge (m1) -- node[text=red]{}(s1);
    \end{naconly}
 	\end{mytikz}
};
\end{mytikz}
}
\quad
\scalebox{.8}{
\begin{mytikz}
\graphsingle (g1) {
\graphname{shape $\mathit{OUT}$:}
    \begin{mytikz}  
	\single (s1) at (0,0) {}; 
    \begin{naconly}
	\single (m1) at (1,0) {};   
    \edge (s1) -- node[text=red]{}(m1);
    \end{naconly}
 	\end{mytikz}
};
\end{mytikz}
}
\quad
\scalebox{.8}{
\begin{mytikz}
\graphsingle (g1) {
\graphname{shape $\mathit{E}$:} 
    \begin{mytikz}  
	\single (s1) at (0,0) {}; 
	\single (s2) at (1,0) {}; 
    \begin{naconly}
    \edge (s1) -- node[text=red]{}(s2);
    \end{naconly}
 	\end{mytikz}
};
\end{mytikz}
}
\quad
\scalebox{.8}{
\begin{mytikz}
\graphsingle (g1) {
\graphname{shape $L$:}
    \begin{mytikz}  
	\single (s1) at (0,0) {}; 
    \begin{naconly}
    \path (s1) edge[loop right,red] node[text=red]{}(s1);
    \end{naconly}
 	\end{mytikz}
};
\end{mytikz}
}
\quad
\scalebox{.8}{
\begin{mytikz}
\graphsingle (g1) {
\graphname{shape $N$:}
    \begin{mytikz}  
   \begin{naconly}
	\single (m1) at (0,0) {};   
    \end{naconly}
 	\end{mytikz}
};
\end{mytikz}
}
\quad
\scalebox{.8}{
\begin{mytikz}
\graphsingle (g1) {
\graphname{shape $\mathit{NL}$:}
    \begin{mytikz}  
   \begin{naconly}
	\single (m1) at (0,0) {};  
	   \path (m1) edge[loop right,red] node[text=red]{}(m1); 
    \end{naconly}
 	\end{mytikz}
};
\end{mytikz}
}
\end{equation}

\noindent

\medskip
\noindent
\emph{Grammars and derivations.} A \emph{(typed) graph grammar (GG)}  $\mathcal G = \langle TG, G_{in}, P, \pi\rangle $ consists of a type graph $TG$, a $TG$-typed input graph $G_{in}$, a set of rule names $P$ and a function $\pi$ assigning to each $p \in P$ a rule $\pi(p) = (L_p \lrel {} K_p \rrel {} R_p)$. 

A \emph{conditional graph grammar (CGG)}
$\mathcal{CG} = \langle TG, G_{in}, P, \pi, \nacs\rangle $ adds to its \emph{underlying graph grammar} $\langle TG, G_{in}, P, \pi \rangle$  a function $\nacs$ 
providing for each $p \in P$ a \emph{negative application condition (NAC)} $\nacs(p)$  over  $\pi(p)$,  that is a set of constraints over $\pi(p)$.  Given a rule\footnote{For the sake of simplicity, we often identify a rule with its name, leaving the application of  $\pi$ implicit.}  $p \in P$ and a match $m: L \to G$, there is a \emph{conditional transformation} $G \Rrel{p,m}  H$ if the match $m: L \to G$ satisfies all the constraints in $\nacs(p)$ and a  double-pushout diagram such as~\eqref{dia:derivation-diagram-def} can be constructed. Sometimes we will call the pair $\langle p, \nacs(p)\rangle$ a \emph{conditional rule}. 


A \emph{derivation} in $\mathcal G$ is a finite sequence of transformations $s =
(G_{in} = G_0 \Rrel{p_1, m_1} \cdots \Rrel{p_n, m_n} G_n)$ with 
$p_i \in P$.   A \emph{conditional derivation} in
$\mathcal{CG}$ is defined similarly as a sequence of \emph{conditional} transformations. A graph $H \in \mathbf{Graphs}_{TG}$ is \emph{reachable} in a (C)GG if there exists a (conditional) derivation ending with $H$. 

A (conditional) graph grammar is \emph{safe} if all rules (including their constraints, if any) and all reachable graphs have an injective typing morphism to $TG$. A simple safe conditional grammar is presented in Example~\ref{ex:Client-Server}. In a safe grammar, since w.l.o.g.~we can consider all typing morphisms as inclusions, we can assume that $G_{in}$ and all reachable graphs are subgraphs of $TG$. Even if safe grammars enjoy quite a limited expressive power, certain variants of them (\emph{occurrence} grammars) are exploited as semantic domain able to represent, through an \emph{unfolding construction},
 causality and independence among transitions of (conditional) grammars, as well as the branching structure of their computations~\cite{BaldanCMR07,corradini2019unfolding}. 

\input{simple-example}

\begin{example}[Client-Server]
\label{ex:Client-Server}
The figure above
shows a safe grammar which depicts a simple Client-Server model. Rather than providing typing in the usual sense, in a safe grammar such as this one the type graph $TG$ plays the role of a global name space for all graphs reachable from the start graph. This is common in models where we have an upper limit on the number of nodes and edges that may exist during the lifetime of the system, or where the grammar represents a semantic object, such as an occurrence grammar obtained by unfolding where rules represent transformation occurrences.

We use an integrated notation 
merging left- and right- hand side graphs into a single rule graph $L \cup R$. We indicate by colours and labels which elements are required and preserved (black), required but deleted (blue), newly created (green), and forbidden (red). That means, the left-hand side $L$ is given by all black and blue elements with the forbidden elements of $N$ in red, the right-hand side by all black and green ones, and the interface $K$ by the black elements only.\footnote{This notation could be ambiguous if a rule has more than one NAC: in that case an additional explanation is needed.} A typical scenario is as follows: initially there are three clients, one for each type. Two of the clients can be promoted to a server if they are not attached to a meeting (using rules $pc(C1)$ and $pc(C2)$ respectively which have a constraint of shape $OUT$). A server can start a meeting if it doesn't already have one (using rules $sm(S1)$ and $sm(S2)$ respectively which have a constraint of shape $IN$). A client can join a meeting if it is not already in that meeting (using rules $jm(C1,M2)$, $jm(C2,M1)$, $jm(C3,M1)$, and $jm(C3,M2)$ respectively which have a constraint of shape $E$). 
\end{example}

\begin{equation}
\label{eq:seq-independence} 
 \xymatrix@R=1.5pc@C=1.0cm{ 
  {N_{1}} & & & & {N_{2}}\\
      {{L_1}}  \ar[u]^{n_{1}}  \ar@{->}[d]_{{m_1}}
     & {{K_1}}  \ar@{->}[l]_{{l_1}} \ar@{->}[r]^{{r_1}} \ar[d]_{n_1} 
     & {{R_1}}  \ar[dr]_(.5){{m^{*}_1}} \ar@{-->}@/^0.3cm/[drrr]^(.7){i} 
   & & {{L_2}}  \ar[u]^{n_{2}}    \ar@{->}[dl]^(.6){{m_2}} \ar@{-->}@/_0.3cm/[dlll]_(.7)j 
     & {{K_2}}  \ar@{->}[l]_{{l_2}} \ar@{->}[r]^{{r_2}} \ar[d]^{n_2} 
     & {{R_2}} \ar[d]^{{m^{*}_2}}
 \\  
     {G} 
     & {D_1} \ar[l]^{g_1} \ar[rr]_{h_1}
     &  
     & {{H_{1}}} & 
     & {D_2} \ar[ll]^{g_2} \ar[r]_{h_2}
     & {H_2} \\
 }
 \end{equation}
 
\noindent
\emph{Independence and causality among transformations.}
In the analysis of systems modeled as graph grammars, it is often important to identify when two consecutive (conditional) transformations of a derivation are \emph{sequentially independent}, in the sense that their order can be switched obtaining an ``equivalent'' derivation, and when instead there is a causal relationship among them. For plain transformations, sequential independence of two consecutive transformations like those in Diagram~\eqref{eq:seq-independence} 
(ignoring constraints $n_{1}$ and $n_{2}$)  is characterized (see e.g.~\cite{corradini1997algebraic}) by requiring the existence of two morphisms $i:R_{1}\to D_{2}$ and $j: L_{2} \to D_{1}$ such that $m_{1}^{*} = g_{2} \circ i $ and $m_{2 }= h_{1} \circ j$.
For the conditional case  \cite{lambers2009certifying} additionally it is required that  $m'_{2} = g_{1} \circ j \models n_{2}$ for each $n_{2} \in \nacs(p_{2})$, and that the match of $p_{1}$ in $H'_{1}$ induced by the transformation $G \Rrel{p_{2},m'_{2}} H'_{1}$   satisfies each $n_{1} \in \nacs(p_{1})$.

%


\subsection{Morphisms of Conditional Grammars}
\label{sec:grammar-morphisms}
The conditional grammar morphisms that we need in this paper are more general than, e.g., those of~\cite{corradini2019unfolding}, because we want to relate a grammar where NACs are encoded by adding complemented items to the type graph of the original grammar. Thus retyping the graphs along the morphism should allow to delete part of the structure. This can be obtained by relating the type graph of the source and the target grammar with a span, as e.g.~in~\cite{DBLP:conf/tagt/BaldanCM98}. Furthermore rules can be mapped to ancestor rules 
and constraints are reflected only up to subsumption.

A morphism $g: TG_0 \to TG_1$ induces a pullback functor $g^<: \mathbf{Graphs}_{TG_1} \to \mathbf{Graphs}_{TG_0}$
mapping an instance graph $\langle G_1, t_{G_1}\rangle$ over $TG_1$ to  the graph  $\langle G_0, t_{G_0}\rangle$ obtained as the pullback $TG_0 \lrel {t_{G_0}} G_0 \to G_1$ of $t_{G_1}$ and $g$. 
Dually, 
a morphism $h: TG_0 \to TG_2$ between type graphs induces a translation of instances obtained by the covariant retyping functor $h^>: \mathbf{Graphs}_{TG_0} \to \mathbf{Graphs}_{TG_2}$, defined by post-composition of $h$ with the typing morphism.
Note that this functor only affects the typing, 
hence $h^>$  acts as identity on morphisms.
Combining both actions, given a span $f = (TG_1 \lrel g TG_0 \rrel h TG_2)$ relating $TG_1$ and $TG_2$ we define functor $f^{<>}: \mathbf{Graphs}_{TG_1} \to \mathbf{Graphs}_{TG_2}$ by the composition $h^> \circ g^<$.

\begin{definition}[Conditional grammar morphisms]
\label{def:cond-morphisms}
Let $\mathcal{CG}_{i} = \langle  TG_i, G_{in,i},$ $ P_i, \pi_i, \nacs_i \rangle$ for $i \in \{1,2\}$ be conditional graph grammars. A \emph{CGG morphism} $f: \mathcal{CG}_{1} \rightarrow
\mathcal{CG}_{2}$ 
is a pair $\langle f_{TG}, f_P \rangle$ where 
$f_{TG} = (TG_1 \lrel {f_<} TG_0 \rrel{f_>} TG_2)$
is a span between type graphs $TG_1$ and $TG_2$ with $f_{<}$ mono,   and $f_P: P_1 \rightarrow P_2$ is a mapping of rule names such that 
%
%
%
\begin{enumerate}
\item The input graph is preserved (up to retyping): $f_{TG}^{<>}(G_{in,1}) = G_{in,2}$

\noindent
\begin{minipage}[c]{.5\linewidth}
\item Rules are mapped to ancestor rules: For all $p \in P_1$, if $f_P(p) = p'$ and $\pi_2(p') = \langle L' \leftarrow K' \rightarrow R'\rangle$, then there is a  DPO diagram like~\eqref{dia:morphism-prod} where the vertical morphisms are mono.
\end{minipage}
\begin{minipage}[c]{.49\linewidth}
    \vspace{-2ex}
    \begin{align}
      \label{dia:morphism-prod}
      \begin{array}[c]{l}
        \xymatrix@R=4ex@C=4ex{
          L' 
          \ar[d]_{i_{p}}
          & 
          K'
          \ar[l]
          \ar[r]
          \ar[d]
          &
          R'
          \ar[d]
          \\
          f_{TG}^{<>}(L_{p})
          &
          f_{TG}^{<>}(K_{p})
          \ar[l]
          \ar[r]
          &
          f_{TG}^{<>}(R_{p})
        }
      \end{array}
    \end{align}
  \end{minipage}



\noindent
\begin{minipage}[c]{.5\linewidth}
\item Constraints are reflected: For each rule $p\in P_{1}$ (with $f_P(p) = p'$) and each mono $h: L_{p} \to N$ over $TG_1$, if there is a constraint $n'  \in \nacs_2(p')$ such that square~\eqref{dia:nac} is a pushout, then there is a constraint $n \in \nacs_1(p)$ such that $n \models h$.
\end{minipage}
\begin{minipage}[c]{.49\linewidth}
    \vspace{-2ex}
    \begin{align}
      \label{dia:nac}
      \begin{array}[c]{l}
        \xymatrix@R=4ex@C=6ex{
          L' 
          \ar[d]_{i_{p}} \ar[r]^{n'}
          & 
	N' 
          \ar[d]
          \\
          f_{TG}^{<>}(L_{p}) \ar[r]^{f_{TG}^{<>}(h)} 
          &
          f_{TG}^{<>}(N)
}      \end{array}
    \end{align}
  \end{minipage}
  

\end{enumerate}
\end{definition}

The definition of conditional grammar morphisms guarantees that morphisms preserve derivations: this will be pivotal for the results presented later. 

%

\begin{proposition}
\label{pr:CGGmorphisms}
Let $\langle f_{TG}, f_P \rangle: \mathcal{CG}_{1} \rightarrow
\mathcal{CG}_{2}$ be a CGG morphism. 
Then for each derivation $G_{in,1} = G_0 \Rrel{p_1} \cdots \Rrel{p_n} G_n$ of  $\mathcal{CG}_{1}$ there is a derivation  $G_{in,2} = f_{TG}^{<>}(G_0) \Rrel{f_P(p_1)} \cdots \Rrel{f_P(p_n)} f_{TG}^{<>}(G_n)$ of $\mathcal{CG}_{2}$.
\end{proposition}

\begin{proof}
Let us denote by ${f}^{<}: \mathbf{Graphs}_{TG_{1}} \to \mathbf{Graphs}_{TG_{0}}$ the pullback functor induced by $f_{<}: TG_{0} \to TG_{1}$,  by ${f}^{>}:  \mathbf{Graphs}_{TG_{0}} \to \mathbf{Graphs}_{TG_{2}}$ the functor obtained by post-composition with $f_{>} :TG_{0} \to TG_{2}$, and by $\mathbf{f}$ their composition.

 We show that if  $G \Rrel{p, m} H$ is a conditional transformation in $\mathcal{CG}_1$ then  $\mathbf{f}(G) \Rrel{p', \mathbf{f}(m)\circ i_{p}} \mathbf{f}(H)$
 is a conditional transformation in $\mathcal{CG}_2$, where $p' = f_P(p)$ and $i_{p}: L_{p'} \to \mathbf{f}(L_{p})$ is as in Diagram~\eqref{dia:morphism-prod}. This fact extends to arbitrary derivations by concatenation and by observing that $\mathbf{f}(G_{in,1}) = G_{in,2}$.
For the DPO diagram of transformation
$G \Rrel{p, m} H$ it is sufficient to observe that pushouts are preserved both by
${f}^{<}$, because $\mathbf{Graphs}$ is an adhesive category \cite{lack2004adhesive},
and by functor ${f}^{>}$ because colimits are computed pointwise in slice categories. Therefore we obtain the required DPO diagram witnessing $\mathbf{f}(G) \Rrel{p', \mathbf{f}(m)\circ i_{p}} \mathbf{f}(H)$ by composing the DPO of Diagram~\eqref{dia:morphism-prod} with the image  via $\mathbf{f}$ of $G \Rrel{p, m} H$. 

    \begin{equation}
   \label{eq:reflectingNacs}
\xymatrix@R=4ex@C=5ex{
 & N \ar@{}[rrd]|{\cycl{2}}   \ar@{..>}@/_0.4cm/[ddl]_{\color{gray} t_{N}} \ar@{-->}@/_0.6cm/[ddd]^{q'} 
 & & X  \ar@{..>}@/_0.4cm/[ddl]\ar@{..>}@/^0.3cm/[ddr] \ar@{-->}[ll]_{} 
 & N'  \ar@{-->}[l] \ar@/^1.5cm/[dddl]^{q} \ar@{..>}@/^0.3cm/[dd]  \ar@{}[dl]|(.4){\cycl{1}}
 \\
 & L_{p} \ar[dd]^(.7){m} \ar@{-->}[u]_{h} \ar@{..>}@/_0.2cm/[dl]
 & & \mathbf{f}(L_{p}) \ar[dd]^(.7){\mathbf{f}(m)} \ar@{-->}[u]_{x}  \ar@{->}[ll]  \ar@{..>}@/_0.2cm/[dl]\ar@{..>}@/^0.2cm/[dr] 
 & L_{p'} \ar[l]_{i_{p}} \ar[u]^{n'} \ar@{..>}[d]
 \\
 {\color{gray}TG_{1}} 
 & & {\color{gray} TG_{0} \ar@{..>}[ll]_(.3){ \color{gray}f_{<}} \ar@{..>}[rr]^(.7){ \color{gray}f_{>}} }
 &  & { \color{gray}TG_{2}}
 \\
 & G \ar@{..>}@/^0.3cm/[ul] 
 & &  \mathbf{f}(G) \ar@{..>}@/_0.3cm/[ur] \ar@/^0.3cm/@{..>}[ul]  \ar@{->}[ll]_{f^{*}_{G}}
}      
 \end{equation}
 
It remains to show that match $\mathbf{f}(m)\circ i_{p}$ satisfies all the constraints in $\mathcal{N}_{2}(p')$, assuming that $m$ satisfies all those in $\mathcal{N}_{1}(p)$. We proceed by contradiction, assuming that there is a constraint $n': L_{p'} \to N' \in \nacs_2(p')$  not satisfied by $\mathbf{f}(m)\circ i_{p}$, i.e., there exists an injective morphism $q: N' \to \mathbf{f}(G)$ such that $q \circ n' = \mathbf{f}(m) \circ i_{p}$. In Diagram~\eqref{eq:reflectingNacs} we show the relevant graphs and morphisms in $\mathbf{Graphs}$, with the corresponding typing morphisms.

Since $q \circ n' = \mathbf{f}(m) \circ i_{p}$ commutes in $\mathbf{Graphs}_{TG_{2}}$, it also commutes in $\mathbf{Graphs}$; also, its typing to $TG_{2}$ factorizes through $TG_{0}$: for $\mathbf{f}(m)$ by the definition of functor $\mathbf{f}$, and for $L_{p'}$ and $N'$ by composing $q$ with the typing of $\mathbf{f}(G)$. 
Let squares \cycl{1} and \cycl{2} be built as pushouts in $\mathbf{Graphs}$.
%
By the pushout property, since the relevant squares are easily shown to commute, there are unique morphisms 
$t_{N}: N \to TG_{1}$ (such that $L_{p} \to TG_{1} = t_{N}\circ h$, showing that $h$ is in $\mathbf{Graphs}_{TG_{1}}$) and  $q': N \to G$,  such that $(\dagger)~m = q' \circ h$.     

Recalling that $\mathbf{Graphs}$ is adhesive, observe that $h$ is mono (because $n'$ is mono and pushouts preserve monos) and that $\mathbf{f}(h) = x$ (up to an iso that we can safely ignore, because pushouts along monos are pullbacks).  Since \cycl{1} is a pushout, condition 3 of the definition of morphism applies, thus there is a constraint ${n} \in \nacs_{1}(p)$ such that ${n} \models h$. But $(\dagger)$  shows that $m \not \models h$ and thus $m \not \models n$, contradicting the assumption. 
%
 \end{proof}

Recall that that sequential independence of conditional transformations is defined in terms of (1) existence of certain morphisms between given graphs and commutativity requirements, and (2) satisfaction of certain NACs by suitable matches. Properties (1) are easily shown to be preserved by any functor, and Proposition~\ref{pr:CGGmorphisms} guarantees that  (2)  is preserved by grammar morphisms. Therefore grammar morphisms preserve sequential independence.
   
   \begin {corollary}
Let $\langle f_{TG}, f_P \rangle: \mathcal{CG}_{1} \rightarrow
\mathcal{CG}_{2}$ be a CGG morphism. 
If $G_0 \Rrel{p_1} G_{1} \Rrel{p_2} G_{2}$ are sequential independent transformations 
 for  $\mathcal{CG}_{1}$ then $f_{TG}^{<>}(G_0) \Rrel{f_P(p_1)} f_{TG}^{<>}(G_1) \Rrel{f_P(p_2)} f_{TG}^{<>}(G_2)$  are sequential independent transformations for $\mathcal{CG}_{2}$.
%
%
\end{corollary}

%% file: simple-example.tex
\begin{center}
\scalebox{.75}{ 
\begin{tabular}{c}
\begin{mytikz}
\graphsingle (g1) {
\graphname{type graph $TG$:}
    \begin{mytikz}  
\usesingle (c1) at (0,0) {$C1$};
\usesingle (c2) at (0,-2) {$C2$};
\usesingle (c3) at (0,-1) {$C3$};
\usesingle (m2) at (2,0) {$M2$};
\usesingle (m1) at (2,-2) {$M1$};
 \usesingle (s2) at (4,0) {$S2$};
 \usesingle (s1) at (4,-2) {$S1$};
\useedge (c1) edge[bend left=-0]  node[uselab]{$in12$}(m2);
 \useedge (c2) edge[bend left=-0]  node[uselab]{$in21$}(m1);
 \useedge (c3) edge[bend left=-0]  node[uselab]{$in32$}(m2);
 \useedge (c3) edge[bend left=-0]  node[uselab]{$in31$}(m1);
 \useedge (m1) edge[bend left=-0]  node[uselab]{$by1$}(s1);
 \useedge (m2) edge[bend left=-0]  node[uselab]{$by2$}(s2);
		\end{mytikz}
};
\end{mytikz}
~~
 \begin{mytikz}
\graphsingle (sg) {
\graphname{start graph $G_{0}$:}
    \begin{mytikz} 	
	\single (c1) at (0,0) {:$C1$};  
	\single (c2) at (0,-1) {:$C2$};  
	\single (c3) at (0,-2) {:$C3$}; 
 	\end{mytikz}
};
\end{mytikz}
~~~
\begin{mytikz}
\graphsingle (g1) {
\graphname{rule $pc(C1)$:}
    \begin{mytikz}  
    \begin{lhsonly}
	\single (c1) at (0,0) {:$C1$};
	\end{lhsonly}
    \begin{naconly}
	\single (m2) at (0,1.5) {:$M2$};   
    \edge (c1) -- node[uselab, text=red]{:$in12$}(m2);
    \end{naconly}
    \begin{rhsonly}
    	\single (s1) at (1,1) {:$S1$}; 
    \end{rhsonly}
 	\end{mytikz}
};
\end{mytikz}
~~
\begin{mytikz}
\graphsingle (g1) {
\graphname{rule $pc(C2)$:}
    \begin{mytikz}  
    \begin{lhsonly}
	\single (c2) at (0,0) {:$C2$};
	\end{lhsonly}
    \begin{naconly}
	\single (m1) at (0,1.5) {:$M1$};   
    \edge (c2) -- node[uselab, text=red]{:$in21$}(m1);
    \end{naconly}
    \begin{rhsonly}
    	\single (s2) at (1,1) {:$S2$}; 
    \end{rhsonly}
 	\end{mytikz}
};
\end{mytikz}
~~~
\begin{mytikz}
\graphsingle (g1) {
\graphname{rule $sm(S1)$:}
    \begin{mytikz}  
	\single (s1) at (0,0) {:$S1$}; 
    \begin{naconly}
	\single (m1) at (0,1.25) {:$M1$};   
    \edge (m1) -- node[uselab, text=red]{:$by1$}(s1);
    \end{naconly}
    \begin{rhsonly}
	\single (m1) at (0,-1.25) {:$M1$};   
    \edge (m1) -- node[uselab,text=green!70!black]{:$by1$}(s1);    	
    \end{rhsonly}
 	\end{mytikz}
};
\end{mytikz}
~~
\begin{mytikz}
\graphsingle (g1) {
\graphname{rule $sm(S2)$:}
    \begin{mytikz}  
	\single (s2) at (0,0) {:$S2$}; 
    \begin{naconly}
	\single (m2) at (0,1.25) {:$M2$};   
    \edge (m2) -- node[uselab, text=red]{:$by2$}(s2);
    \end{naconly}
    \begin{rhsonly}
	\single (m2) at (0,-1.25) {:$M2$};   
    \edge (m2) -- node[uselab,text=green!70!black]{:$by2$}(s2);    	
    \end{rhsonly}
 	\end{mytikz}
};
\end{mytikz}
\\
\begin{mytikz}
\graphsingle (g1) {
\graphname{rule $jm(C1,M2)$:}
    \begin{mytikz}  
	\single (c1) at (0,0) {:$C1$}; 
	\single (m2) at (0,-1.5) {:$M2$}; 
    \begin{naconly}
    \edge (c1) -- node[uselab, text=red]{:$in12$}(m2);
    \end{naconly}
    \begin{rhsonly}
    \edge[-] (c1.east) edge[->,bend left=45] node[uselab,text=green!70!black]{:$in12$}(m2.east);    	
    \end{rhsonly}
 	\end{mytikz}
};
\end{mytikz}
\quad
\begin{mytikz}
\graphsingle (g1) {
\graphname{rule $jm(C2,M1)$:}
    \begin{mytikz}  
	\single (c2) at (0,0) {:$C2$}; 
	\single (m1) at (0,-1.5) {:$M1$}; 
    \begin{naconly}
    \edge (c2) -- node[uselab, text=red]{:$in21$}(m1);
    \end{naconly}
    \begin{rhsonly}
    \edge[-] (c2.east) edge[->,bend left=45] node[uselab,text=green!70!black]{:$in21$}(m1.east);    	
    \end{rhsonly}
 	\end{mytikz}
};
\end{mytikz}
\qquad
\begin{mytikz}
\graphsingle (g1) {
\graphname{rule $jm(C3,M1)$:}
    \begin{mytikz}  
	\single (c3) at (0,0) {:$C3$}; 
	\single (m1) at (0,-1.5) {:$M1$}; 
    \begin{naconly}
    \edge (c3) -- node[uselab, text=red]{:$in31$}(m1);
    \end{naconly}
    \begin{rhsonly}
    \edge[-] (c3.east) edge[->,bend left=45] node[uselab,text=green!70!black]{:$in31$}(m1.east);    	
    \end{rhsonly}
 	\end{mytikz}
};
\end{mytikz}
\quad
\begin{mytikz}
\graphsingle (g1) {
\graphname{rule $jm(C3,M2)$:}
    \begin{mytikz}  
	\single (c3) at (0,0) {:$C3$}; 
	\single (m2) at (0,-1.5) {:$M2$}; 
    \begin{naconly}
    \edge (c3) -- node[uselab, text=red]{:$in32$}(m2);
    \end{naconly}
    \begin{rhsonly}
    \edge[-] (c3.east) edge[->,bend left=45] node[uselab,text=green!70!black]{:$in32$}(m2.east);    	
    \end{rhsonly}
 	\end{mytikz}
};
\end{mytikz}
\end{tabular}}
\end{center}
\vspace{.1cm}

%% file: 03-Encoding.tex

\section{Encoding NACs for Safe Conditional Grammars}
\label{sec:SafeCase}
\label{sec:encoding-safe}

Following the outline sketched in the Introduction,  
given a conditional grammar $\mathcal {CG}$ we define now a grammar $E(\mathcal {CG})$ where the NACs are encoded in the reachable graphs using suitable \emph{complemented items}. Grammar $E(\mathcal {CG})$ keeps essentially the same NACs as $\mathcal {CG}$,  but they will be shown to be redundant because of their encoding. For this section 
let $\mathcal{CG} = \langle TG, G_{in}, P, \pi, \nacs\rangle$ be a fixed \emph{safe}  conditional grammar, where all constraints are \emph{incremental}. 
The first step consists of building an \emph{enriched type graph} $\overline{TG}$ which extends $TG$ by adding {complemented items} for each constraint: intuitively, the presence of such items in a graph guarantees that the constraint is satisfied.

\begin{definition}[enriched type graph]
\label{def:complementedTG}
Let
$sh(\mathcal{CG}) = \{n^{-}: L^{-}_{p} \to N^{-} \mid 
p \in P, n \in \nacs(p) \}$ be the set of all shapes of constraints appearing in rules of $\mathcal{CG}$.
Consider the diagram in $\mathbf{Graphs}$ made of all arrows of $sh(\mathcal{CG})$, of $TG$, and for each shape $n^{-}: L^{-}_{p} \to N^{-} $
of the typing morphism $t_{L_{p}^{-}}: L_{p}^{-} \to TG$. Then the \emph{enriched type graph} $\ov{TG}$ is the colimit of this diagram.
\end{definition}

\noindent
\begin{minipage}{.59\textwidth}
Intuitively,  the colimit adds to $TG$  for each shape a copy of all the red items (see Diagram~\eqref{eq:shapes}), their \emph{complement}, typed over $\ov{TG}$.  
Diagram~\eqref{eq:compl} shows that the body  $N^{-}$ of  shape $n^{-}$ has two monos to $\ov{TG}$:  the colimit injection $\ov{t}_{N^{-}}$ as well as $in_{TG} \circ t_{N^{-}}$.    These monos are different (because the shape is not an iso by assumption) but coincide on the border $L^{-}_{p}$.   
\end{minipage}
\begin{minipage}{.4\textwidth}
\vspace{-.7cm}
\begin{equation}
\label{eq:compl}
\xymatrix@R=3ex{
&  N \ar@{-->}@/_0.4cm/[ddl]_{\ov{t}_{N}}  \ar@{..>}@/^0.4cm/[ddr]^(.4){t_{N}}
& &  L_{p} \ar@{..>}[ddl] \ar[ll]_{n} \\
& N^{-} \ar[u] \ar@{-->}@/_0.3cm/[dl]^{\ov{t}_{N^{-}}} \ar@{..>}[dr]|{t_{N^{-}}}
& & L^{-}_{p}  \ar@{..>}[dl]^{t_{L^{-}_{p}}}  \ar[ll]_{n^{-}} \ar[u] \\
\ov{TG} \ar@{}[rrruu]|(.3){ }& & TG \ar@{-->}[ll]_{in_{TG}}
}
\end{equation}
\end{minipage}
They determine two different objects of $\mathbf{Graphs}_{\ov{TG}}$ that we shall denote $N^{-}$ ($= \langle N^{-},  in_{TG} \circ t_{N^{-}}\rangle$) and $\ov{N}^{-}$ ($= \langle N^{-},  \ov{t}_{N^{-}}\rangle$), respectively, calling the second \emph{the body's complement}.
Similarly, since the square in~\eqref{eq:compl} is a pushout, there is a mediating morphism $\ov{t}_{n}$ from  $N$  to $\ov{TG}$ (which is a mono by adhesivity), besides the mono $in_{TG} \circ t_{N}$, giving rise to two objects in $\mathbf{Graphs}_{\ov{TG}}$ that we will denote $\ov{N}$ and $N$, respectively. We call $\ov{n}: L_{p} \to \ov{N}$ the \emph{complemented constraint} of $n$.

We now introduce  a property on obiects of $\mathbf{Graphs}_{\ov{TG}}$, called the \emph{complementation invariant}.
Intuitively, a graph $G$ over  $\ov{TG}$ satisfies the invariant if for each constraint $n: L \to N$ of the original grammar, whenever the border of its shape ($L^{-}$) is present in $G$, then  the body $N^{-}$ of the shape  
is present as well if and only if  its complement $\ov{N}^{-}$ 
is not present.


\begin{definition}[complementation invariant]
\label{def:invariant}
Let  $sh(\mathcal{CG})$ and $\ov{TG}$ be as in Definition~\ref{def:complementedTG}.
Then, using a syntax reminiscent of \emph{nested application conditions} \cite{DBLP:journals/mscs/HabelP09}, the \emph{complementation invariant}
is defined as\quad
$\Phi_{inv} = \forall (n^{-}: L^{-} \to N^{-}) \in sh(\mathcal{CG}) \,.\, \forall L^{-} \,.\, \exists (L^{-} \to N^{-})~\mbox{\textsc{xor}}~\exists (L^{-} \to \ov{N}^{-})$
\end{definition}

Given a graph $G$ over $TG$ we can enrich it with the needed complemented items in order to obtain a graph $\inv{G}$ over $\ov{TG}$ that satisfies the invariant.

%

\begin{definition}[invariant closure]
\label{def:invariantClosure}
Given a graph $G$ typed over $TG$, its \emph{invariant closure} $\inv{G}$ is the graph typed over $\ov{TG}$ obtained as the colimit of the diagram including (a) graph $G$ typed over $\ov{TG}$ by $in_{TG}\circ t_{G}$, (b) for each shape  $n^{-}: L^{-} \to {N}^{-} \in sh(\mathcal{CG})$, the complemented shape $n^{-}: L^{-} \to \ov{N}^{-}$ and the inclusion $L^{-} \to G$ if and only if the border $ L^{-}$ is contained in $G$ but the body $N^{-}$ is not.
\end{definition}

It is easy to check that for each graph $G$ over $TG$, $in_{TG}^{<}(\inv{G}) \cong G$. Thus $\inv{\_}$ is an inverse to the object component of the retyping functor along  $in_{TG}$, but it is not itself a functor.

\begin{example}[Complemented Client-Server]
\label{ex:complementedCS}
The grammar of the figure below
is obtained by applying to the Client-Server grammar of Example~\ref{ex:Client-Server} the transformations  described in this section. The NACs are not represented (they are redundant, as we will see). The type graph $\ov{TG}$ and the start graph $\ov{G}_{0} = \inv{G_{0}}$ are obtained   according to Defs.~\ref{def:complementedTG} and~\ref{def:invariantClosure}. The  rules will be illustrated later on.
\end{example}

\input{simple-example-complementation}

\noindent
As examples, the invariants for the constraints of rules $pc(C1)$ and $jm(C1,M2)$ of Example~\ref{ex:Client-Server} (which coincide with their shapes) are the following:
 \begin{center}
 \begin{tabular}{*{8}{l}}
 invariant for ${pc(C1)}$: &
\begin{mytikz}
\node (n1) {$\forall~$};
\end{mytikz}
&
\scalebox{.6}{
\begin{mytikz}
\graphsingle (g1) [right=.1cm of n1] {
 \begin{mytikz}  
	\single (c1) at (0,0) {:$C1$};
	\end{mytikz}};
\end{mytikz}}
&
\begin{mytikz}
\node (n2) [right=.1cm of g1] {$.~\exists~\big($};
\end{mytikz}
&
\scalebox{.6}{
\begin{mytikz}
\graphsingle (g2) [right=.1cm of n2] {
    \begin{mytikz}  
	\single (c1) at (0,0) {:$C1$};
	\single (m2) at (2.5,0) {:$M2$};   
    \edge (c1) -- node[uselab]{:$in12$}(m2);
 	\end{mytikz}
};
\end{mytikz}}
&
\begin{mytikz}
\node (xor) [right=.1cm of g2] {$\big)~\mbox{\textsc{xor}}~\exists~\big($};
\end{mytikz}
&
\scalebox{.6}{
\begin{mytikz}
\graphsingle (g2) [right=.1cm of xor] {
    \begin{mytikz}  
	\single (c1) at (0,0) {:$C1$};
			\single (m22) at (2.5,0) {:$\overline{M2}_1$}; 
     \edge (c1) -- node[uselab]{:$\overline{in12}_1$}(m22);
 	\end{mytikz}
};
\end{mytikz}}
&
\begin{mytikz}
\node (n3) [right=.1cm of g2] {$\big)$};
\end{mytikz}
\\
invariant for ${jm(C1,M2)}$: &
\begin{mytikz}
\node (n1) {$\forall~$};
\end{mytikz}
&
\scalebox{.6}{
\begin{mytikz}
\graphsingle (g1) 
{
 \begin{mytikz}  
	\single (c1) at (0,0) {:$C1$};
	\single (m2) at (1,0) {:$M2$};
	\end{mytikz}};
\end{mytikz}
}	
&
\begin{mytikz}
\node (n2) [right=.1cm of g1] {$.~\exists~\big($};
\end{mytikz}
&
\scalebox{.6}{
\begin{mytikz}
\graphsingle (g2) [right=.1cm of n2] {
    \begin{mytikz}  
	\single (c1) at (0,0) {:$C1$};
	\single (m2) at (2.5,0) {:$M2$};   
    \edge (c1) -- node[uselab]{:$in12$}(m2);
 	\end{mytikz}
};
\end{mytikz}}
&
\begin{mytikz}
\node (xor) [right=.1cm of g2]  {$\big)~\mbox{\textsc{xor}}~\exists~\big($};
\end{mytikz}
&
\scalebox{.6}{
\begin{mytikz}
\graphsingle (g2) [right=.1cm of xor] {
    \begin{mytikz}  
	\single (c1) at (0,0) {:$C1$};
			\single (m22) at (2.5,0) {:$M2$}; 
     \edge (c1) -- node[uselab]{:$\overline{in12}$}(m22);
 	\end{mytikz}
};
\end{mytikz}}
&
\begin{mytikz}
\node (n3) [right=.1cm of g2] {$\big)$};
\end{mytikz}\\
 \end{tabular}
  \end{center}
For example a graph satisfies the first invariant if, whenever it  contains node :$C1$, then it contains an outgoing edge $in12$ to node :$M2$ if and only if it does not contain an outgoing edge $\ov{in12_{1}}$ to node :$\ov{M2}_{1}$.

We proceed now describing how the rules of $\mathcal{CG}$ have to be modified exploiting the complemented items in $\ov{TG}$  to encode the constraints in the reachable graphs. We first complement all rules, making them applicable only if the complement of the shape of each constraint is present, which means that the constraint cannot be violated. Next we  enrich the resulting rules to ensure that the invariant is preserved.

\begin{construction}[complemented rules]
\label{def:compProds}
Given a conditional rule $\langle p, \nacs_{p}\rangle$ over $TG$ 
its complementation is the conditional rule $compl(\langle p, \nacs_{p}\rangle)$
over $\ov{TG}$
returned by the following procedure.

\begin{enumerate}
\item Set $\tilde{p} = p$, $\nacs_{\mathit{todo}} = \nacs_{p}$, $\nacs_{\mathit{done}} = \emptyset$.
\item While $\nacs_{\mathit{todo}}$ is not empty, perform the following step:
\begin{itemize}
\item[3. ] Let $n : L \to N \in \nacs_{\mathit{todo}}$ be a constraint for $\tilde{p} = ({L} \lrel{{l}} {K} \to {R})$. We assume that 
$(\dagger)$  the typing over $\ov{TG}$ of the body $N^{-}$ of $n$'s shape 
is $t_{N^{-}};in_{{TG}}$, i.e.~it factorizes through $TG$.
Then 
\begin{enumerate}
\item If there is a DPO diagram $\ov{N} \Rrel{\tilde{p},\ov{n}} H$, where $\ov{n}: L \to \ov{N}$ is the complemented constraint, and there is no arrow $N^{-} \to R$, then set $\tilde{p}$ to be the derived rule (its bottom span). Otherwise, set $\tilde{p} = (\ov{N} \lrel {\ov{n} \circ {l}} {K} \to {R})$.
%
\item Set $\nacs_{\mathit{todo}} = \nacs_{\mathit{todo}} \setminus \{n\}$, and $\nacs_{\mathit{done}} = \nacs_{\mathit{done}} \cup \{n\}$. 
\item Lift each constraint in $\nacs_{\mathit{todo}}$ or $\nacs_{\mathit{done}}$ along 
$\ov{n}$, i.e., replace each constraint $n': L \to N'$ with the morphism obtained by pushing out $n'$  along $\ov{n} : L \to \ov{N}$.
\end{enumerate}
\end{itemize}
\item[4. ]   Return $\langle \tilde{p}, \nacs_{\mathit{done}}\rangle$
\end{enumerate}
\end{construction}

Note that assumption $(\dagger)$ above is satisfied by the constraints of the starting rule $p$ because it is typed over $TG$. Furthermore it is preserved by the modification of constraints in step \emph{3.(c)} because the shape of a constraint does not change when pushing it out along a morphism, as it is defined as an initial pushout. 

Note that Construction~\ref{def:compProds} was phrased exploiting pushouts in category $\mathbf{Graphs}_{\ov{TG}}$. The assumption of safety would allow us to work directly in the category of subgraphs of $\ov{TG}$, using the standard set-theoretical operations of intersection and union of subgraphs (corresponding to pullback over the type graph, and to pushout over the pullback).  We refrain from this and continue exploiting categorical constructions in view of generalizing the theory to the unsafe case.

We enrich now the complemented rules just obtained  in order to guarantee the preservation of the complementation invariant of  Definition~\ref{def:invariant}. Recall that the invariant requires that  whenever the border $L^{-}$ of the shape of a constraint is present, then either its body $N^{-}$ or its body complement $\ov{N}^{-}$ is present as well. For the shapes of constraints of a rule $p$, it is easy to check that such property is preserved, by construction, by the rule $compl(\langle p, \nacs_{p}\rangle)$ returned by Construction~\ref{def:compProds}.  
With the following construction we replace  each rule  with one or more rules (all equivalent as far as the items typed over $TG$ are concerned) which preserve the invariant also for the shapes of constraints of other rules. 

The following  four cases are possible for a rule $p$ and a shape $n^{-}: L^{-} \to N^{-}$ not in $sh(\mathcal{N}_{p})$: (1) If rule $p$ deletes the body $N^{-}$ preserving its border $L^{-}$, then $p$ is enriched  to create the corresponding body complement $\ov{N}^{-}$. (2) Dually, if $p$ creates  the body $N^{-}$ preserving the border $L^{-}$, then it is extended to delete the body complement $\ov{N}^{-}$. (3) If $p$ creates the border $L^{-}$ but not the body $N^{-}$, then it is extended to create the body complement $\ov{N}^{-}$. (4)  Finally, in the case of constraints of shape $E$, the border 
(which is made of two nodes, see Diagram~\eqref{eq:shapes}) can also be created or deleted only in part (just one node). If a rule creates only one node of the border $L^{-}$ (i.e.  the second node is not in $R$), then to preserve the invariant we need two rules: if the other node is already present in the current graph, we need to add the body complement,  but if it is not present we don't need to add anything.


\begin{construction}[making the rules invariant-preserving]
\label{def:rules-inv}
Let $CP_{\mathcal{CG}} = \{compl(\langle p, \nacs(p)\rangle) \mid p \in P\}$ be the set of complemented rules of  grammar $\mathcal{CG}$ 
typed over $\ov{TG}$. Also, for $p \in P$ let $sh(\mathcal{CG}_{\ov{p}})$ be the set of all shapes of constraints of rules different from $p$: as observed above, shapes are not affected by complementation.

 For each conditional rule $\langle p = L \leftarrow K \to R, \nacs_{p}\rangle$ in $CP_{\mathcal{CG}}$  we obtain a set of \emph{invariant preserving} conditional rules $\mathit{IP}(p)$ by applying in sequence the following transformations. 
\begin{enumerate}
\item  \emph{[Compensate body deletion]}\  
%
%
%
For each shape $n^{-}_{i}: L^{-}_{i} \to N^{-}_{i}$ in $sh(\mathcal{CG}_{\ov{p}})$ such that there is an arrow $N^{-}_{i} \to L$, check if in Diagram~\eqref{eq:inv1} (left) the top arrow of pullback  (1) is not an isomorphism (the body $N^{-}_{i}$ is deleted), but the top arrow of pullback (1) + (2) is an iso (the border $L^{-}_{i}$ is preserved).  In this case take the pushout (3) of arrows $L^{-}_{i} \cong Y \to K \to R$ and $Y \to \ov{N}^{-}_{i}$, and set $p' := L \leftarrow K \to R'$. After all shapes are considered, continue to the next step with $\langle p', \nacs_{p}\rangle$.

\begin{equation}
\label{eq:inv1}
\scalebox{.90}{
\xymatrix@C=4ex@R=3ex{
L_{i}^{-} \ar[d] \ar@/_3ex/[dd]
& Y   \ar[l]_{ \cong} \ar[r]  \ar[d] \ar@/^3ex/[dd]
& \ov{N}_{i}^{-} \ar@{-->}[d]  \ar@{}[dl]|{(3)} 
\\ N^{-}_{i} \ar[d] 
& X \ar[l]_{\not \cong} \ar[d] \ar@{}[dl]|{(1)} \ar@{}[ul]|{(2)}
&  R'
\\L 
& K \ar[l] \ar[r]
& R \ar@{-->}[u]
}
\qquad
\xymatrix@C=4ex@R=3ex{
X \ar[r]^{\not \cong} \ar[d]  \ar@{}[dr]|{(4)} 
& L^{-}_{i} \ar[r] \ar[d]  \ar@{}[dr]|{(5)} 
& \ov{N}^{-}_{i} \ar@{-->}[d]
\\ K \ar[r] 
&R \ar@{-->}[r] 
&R'
}
\qquad
\xymatrix@C=4ex@R=3ex{
& Y \ar[r]^{\not \cong} \ar[d]  \ar[r] \ar@{}[dr]|{(6)} 
& X \ar[d] \ar[r]^{\not \cong}  \ar@{}[dr]|{(7)} 
& L^{-}_{i} \ar[r] \ar@{-->}[d]  \ar@{}[dr]|{(8)} \ar@{..>}@/^3ex/[dd]
& \ov{N}^{-}_{i} \ar@{-->}[d]
\\ L \ar[dr]
& K \ar[r] \ar[l] \ar[dr]  \ar@{}[d]|{(9)}
&R \ar@{-->}[r]  \ar@{}[d]|{(10)} \ar@{..>}[dr]
&R' \ar@{-->}[r] 
&R''
\\ & L' 
& K' \ar[l] \ar@{-->}@/^1ex/[urr] 
& \ov{TG}
}
}
\end{equation}

\item \emph{[Compensate body creation]}\  This step can be formalized with a diagram symmetric to the previous one: if $p$ creates the body $N^{-}_{i}$ of a shape and preserves its border $L^{-}_{i}$,   $p$ is enriched in order to delete the body's complement by replacing $L$ with the graph $L'$ obtained as pushout of $L^{-}_{i} \to N^{-}_{i}$ and $L^{-}_{i} \to K \to L$. Also, all constraints have to be lifted along 
$L \to L'$   as in step \emph{(c)} of Construction~\ref{def:compProds}.

\item \emph{[Complete border creation]}\  Continuing with the conditional rule $\langle p, \nacs_{p}\rangle$ resulting from the previous step, for each shape  $n^{-}_{i}: L^{-}_{i} \to N^{-}_{i}$ in $sh(\mathcal{CG}_{\ov{p}})$ consider Diagram~\eqref{eq:inv1} (middle):
if there is an arrow $L^{-}_{i} \to R$ and (4) is a pullback where the top morphism is not an iso (the rule completes the creation of the border), then build the pushout (5) and continue with  $\langle L \leftarrow K \to R',\nacs_{p}\rangle$.
%
%

\item \emph{[Complete partial border creation]}\  
%
%
Let $\mathcal{R} = \{\langle p, \nacs_{p}\rangle\}$ be a set of conditional rules, initialized with  the  rule  resulting from the previous step.
For each  shape $n_{i}^{-}$ of type $E$ in $sh(\mathcal{CG}_{\ov{p}})$, do the following for each rule $p= L \leftarrow K \to R$ in $\mathcal{R}$:
Consider Diagram~\eqref{eq:inv1} (right).  Graph $X$ is the intersection of $R$ and $L^{-}_{i}$, obtained as the pullback of the typing morphisms (because the grammar is safe). If it is not isomorphic to $L^{-}_{i}$ and also the top morphism of pullback (6) is not an iso, then the rule creates part of the border, but not all of it. In this case add to $\mathcal{R}$ rule $\langle L' \leftarrow K' \to R'', \mathcal{N}'\rangle$, obtained as follows: build (7) and (8) as pushouts, (10) as pushout complement, (9) as pushout, while  $\mathcal{N}'$ is obtained by lifting all the constraint of $\mathcal{N}$ along $L \to L'$. 
\end{enumerate}
\end{construction}

Note that the last step could generate for each conditional rule a large set of derived rules needed to cover all the possible situations (presence or absence of nodes of $E$ shapes). 
Whenever $p$ has a match in a graph $G$, only the largest derived rule of the set having a match in $G$ should be applied. This  requirement of maximality of the match will be needed to guarantee correctness.


For example, rule $\ov{jm(C2,M1)}$ of Example~\ref{ex:complementedCS}
is obtained by first complementing rule $jm(C2,M1)$ of Example~\ref{ex:Client-Server},
which adds edge $\ov{in21}$ to
$L$, and then applying  step 2, because the rule creates ${in21}$ which is also the body of the shape of the constraint of rule $pc(C2)$.
In Example~\ref{ex:complementedCS}
for each group of rules generated according to step 4 above 
only the maximal rule is depicted. For example, rule $\ov{sm(S1)}$ also has three subrules, obtained by deleting node $C2$, $C3$ or both from $L$, $K$ and $R$, and the corresponding complemented edges to $M1$ from $R$. Such rules are obtained starting from rule ${sm(S1)}$ of Example~\ref{ex:Client-Server},
which generates a node, $M1$,  belonging to the border of  two shapes.

Exploiting the constructions just described, we can complete the definition of  grammar $E(\mathcal{CG})$.

\begin{definition}[the enriched grammar]
\label{def:encodedGrammar}
Given $\mathcal{CG} = \langle TG, G_{in}, P, \pi, \nacs\rangle$, its enriched grammar is defined as $E(\mathcal{CG}) = \langle \ov{TG}, \inv{G_{in}}, P', \pi', \nacs'\rangle$ where $\ov{TG}$ is as for  Defintion~\ref{def:complementedTG}, $\inv{G_{in}}$ is as for Definition~\ref{def:invariantClosure}, and the conditional rules determined by $P'$, $\pi'$ and $\nacs'$ are obtained from those of   $\mathcal{CG}$ by applying Constructions~\ref{def:compProds} and~\ref{def:rules-inv}. 
\end{definition}

As desired, all the  reachable graphs of grammar $E(\mathcal{CG})$ satisfy the invariant $\Phi_{inv}$.
In fact this is obvious for the start graph by construction, and each of the rules can be shown to preserve the invariant by a detailed analysis of Constructions~\ref{def:compProds} and~\ref{def:rules-inv}, also exploiting the maximality requirement of matches mentioned above.

\begin{fact}[the complementation invariant is satisfied]
\label{pr:invariant}
All the reachable graphs of conditional grammar  $E(\mathcal{CG})$ satisfy the complementation invariant $\Phi_{inv}$.
\end{fact}

%% file: simple-example-complementation.tex
\noindent
\scalebox{.80}{ 
\begin{tabular}{*{4}{c}}
\multirow{2}{*}[2cm]{
\begin{mytikz}
\graphsingle (g1) {
\graphname{type graph $\overline{TG}$:}
    \begin{mytikz}  
\usesingle (c1) at (0,0) {$C1$};
\usesingle (c2) at (0,-3) {$C2$};
\usesingle (c3) at (0,-1.5) {$C3$};
\usesingle (m2) at (2,0) {$M2$};
 \usesingle (m22) at (-1.5,-1.15) {$\overline{M2}_1$};
  \usesingle (m2s2) at (3,-1.15) {$\overline{M2}_2$};
\usesingle (m1) at (2,-3) {$M1$};
  \usesingle (m11) at (-1.5,-1.7) {$\overline{M1}_1$};
 \usesingle (m1s1) at (3,-1.7) {$\overline{M1}_2$};
 \usesingle (s2) at (4,0) {$S2$};
 \usesingle (s1) at (4,-3) {$S1$};
\useedge (c1) edge[bend left=-0]  node[uselab]{$in12$}(m2);
 \useedge (c2) edge[bend left=-0]  node[uselab]{$in21$}(m1);
 \useedge (c3) edge[bend left=-20]  node[uselab]{$in31$}(m1);
 \useedge (c3) edge[bend left=-20]  node[uselab]{$in32$}(m2);
 \useedge (m1) edge[bend left=-0]  node[uselab]{$by1$}(s1);
 \useedge (m2) edge[bend left=-0]  node[uselab]{$by2$}(s2);
 \useedge (c1) edge[bend left=-0]  node[uselab]{$\overline{in12}_1$}(m22);
 \useedge (c2) edge[bend left=-0]  node[uselab]{$\overline{in21}_1$}(m11);
 \useedge (c1) edge[bend left=20]  node[uselab,above]{$\overline{in12}$}(m2);
 \useedge (c2) edge[bend left=-20]  node[uselab,below]{$\overline{in21}$}(m1);
  \useedge (c3) edge[bend left=20,looseness=1]  node[uselab]{$\overline{in32}$}(m2);
 \useedge (c3) edge[bend left=20,looseness=1]  node[uselab]{$\overline{in31}$}(m1);
 \useedge (m1s1) edge[bend left=-0]  node[uselab,left]{$\overline{by1}$}(s1);
 \useedge (m2s2) edge[bend left=-0]  node[uselab,left]{$\overline{by2}$}(s2);
		\end{mytikz}
};
\end{mytikz}}
&
\multirow{2}{*}[2cm]{
 \begin{mytikz}
\graphsingle (sg) {
\graphname{start graph $\overline{G_{0}}$:}
    \begin{mytikz} 	
	\single (c1) at (0,0) {:$C1$};  
	\single (c2) at (0,-3) {:$C2$}; 
	\single (c3) at (0,-1.5) {:$C3$}; 
\usesingle (m22) at (-1.5,-1) {$\overline{M2}_1$};
 \usesingle (m11) at (-1.5,-2) {$\overline{M1}_1$};
    \edge (c1) -- node[uselab]{:$\overline{in12}_1$} (m22);
    \edge (c2) -- node[uselab]{:$\overline{in21}_1$} (m11);
 	\end{mytikz}
};
\end{mytikz}}
&
\begin{mytikz}
\graphsingle (g1) {
\graphname{rule $\overline{pc(C1)}$:}
	\begin{mytikz}  
    \begin{lhsonly}
			\single (m22) at (2,0) {:$\overline{M2}_1$}; 
	\single (c1) at (0,0) {:$C1$};
     \edge (c1) -- node[uselab,text=blue]{:$\overline{in12}_1$}(m22);
	\end{lhsonly}
    \begin{rhsonly}
    	\single (s1) at (0,-1) {:$S1$};
    		\single (m11) at (2,-1) {:$\overline{M1}_2$}; 	 
     \edge (m11) -- node[uselab,text=green!70!black]{:$\overline{by1}$}(s1);
    \end{rhsonly}
 	\end{mytikz}
};
\end{mytikz}
&
\begin{mytikz}
\graphsingle (g1) {
\graphname{rule $\overline{pc(C2)}$:}
	\begin{mytikz}  
    \begin{lhsonly}
			\single (m11) at (2,0) {:$\overline{M1}_1$}; 
	\single (c2) at (0,0) {:$C2$};
     \edge (c2) -- node[uselab,text=blue]{:$\overline{in21}_1$}(m11);
	\end{lhsonly}
    \begin{rhsonly}
    		\single (m22) at (2,-1) {:$\overline{M2}_2$}; 	 
    	\single (s2) at (0,-1) {:$S2$};
     \edge (m22) -- node[uselab,text=green!70!black]{:$\overline{by2}$}(s2);
    \end{rhsonly}
 	\end{mytikz}
};
\end{mytikz}
\\
& & 
\begin{mytikz}
\graphsingle (g1) {
\graphname{rule $\overline{sm(S1)}$:}
    \begin{mytikz}  
	\single (s1) at (0,0) {:$S1$}; 
		\single (c2) at (2,-.75) {:$C2$}; 
			\single (c3) at (2,-1.5) {:$C3$}; 
	
	\begin{lhsonly}
	\single (m11) at (2,0) {:$\overline{M1}_2$}; 
     \edge (m11) -- node[uselab,text=blue]{:$\overline{by1}$}(s1);
	\end{lhsonly}
    \begin{rhsonly}
	\single (m1) at (0,-1.5) {:$M1$};   
    \edge (m1) -- node[uselab,text=green!70!black]{:$by1$}(s1);    	
    \edge (c2) -- node[uselab,text=green!70!black]{:$\overline{in21}$}(m1);  
      \edge (c3) -- node[uselab,text=green!70!black]{:$\overline{in31}$}(m1);  
    \end{rhsonly}
 	\end{mytikz}
};
\end{mytikz}
& 
\begin{mytikz}
\graphsingle (g1) {
\graphname{rule $\overline{sm(S2)}$:}
    \begin{mytikz}  
	\single (s2) at (0,0) {:$S2$}; 
		\single (c1) at (2,-.75) {:$C1$}; 
			\single (c3) at (2,-1.5) {:$C3$}; 
	
	\begin{lhsonly}
	\single (m22) at (2,0) {:$\overline{M2}_2$}; 
     \edge (m22) -- node[uselab,text=blue]{:$\overline{by2}$}(s2);
	\end{lhsonly}
    \begin{rhsonly}
	\single (m2) at (0,-1.5) {:$M2$};   
    \edge (m2) -- node[uselab,text=green!70!black]{:$by2$}(s2);    	
    \edge (c1) -- node[uselab,text=green!70!black]{:$\overline{in12}$}(m2);  
      \edge (c3) -- node[uselab,text=green!70!black]{:$\overline{in32}$}(m2);  
    \end{rhsonly}
 	\end{mytikz}
};
\end{mytikz}
\\
\begin{mytikz}
\graphsingle (g1) {
\graphname{rule $\overline{jm(C1,M2)}$:}
    \begin{mytikz}  
	\single (c1) at (0,0) {:$C1$}; 
	\single (m2) at (0,-1.5) {:$M2$}; 
	\begin{lhsonly}
	\edge (c1) edge[bend left=90] node[uselab,text=blue]{:$\overline{in12}$} (m2);
			\single (m22) at (2,0) {:$\overline{M2}_1$}; 
     \edge (c1) -- node[uselab,text=blue]{:$\overline{in12}_1$}(m22);	
     \end{lhsonly}
    \begin{rhsonly}
    \edge (c1) edge[bend left=0] node[uselab,text=green!70!black]{:$in12$}(m2);    	
    \end{rhsonly}
 	\end{mytikz}
};
\end{mytikz}
&
\begin{mytikz}
\graphsingle (g1) {
\graphname{rule $\overline{jm(C2,M1)}$:}
    \begin{mytikz}  
	\single (c2) at (0,0) {:$C2$}; 
	\single (m1) at (0,-1.5) {:$M1$}; 
	\begin{lhsonly}
	\edge (c2) edge[bend left=90] node[uselab,text=blue]{:$\overline{in21}$} (m1);
	\single (m22) at (2,0) {:$\overline{M1}_1$}; 
     \edge (c2) -- node[uselab,text=blue]{:$\overline{in21}_1$}(m22);
	\end{lhsonly}
    \begin{rhsonly}
    \edge (c2) edge[bend left=0] node[uselab,text=green!70!black]{:$in21$}(m1);    	
    \end{rhsonly}
 	\end{mytikz}
};
\end{mytikz}
&
\begin{mytikz}
\graphsingle (g1) {
\graphname{rule $\overline{jm(C3,M1)}$:}
    \begin{mytikz}  
	\single (c3) at (0,0) {:$C3$}; 
	\single (m1) at (0,-1.5) {:$M1$}; 
	\begin{lhsonly}
	\edge (c3) edge[bend left=90] node[uselab,text=blue]{:$\overline{in31}$} (m1);
	\end{lhsonly}
    \begin{rhsonly}
    \edge (c3) edge[bend left=0] node[uselab,text=green!70!black]{:$in31$}(m1);    	
    \end{rhsonly}
 	\end{mytikz}
};
\end{mytikz}
&
\begin{mytikz}
\graphsingle (g1) {
\graphname{rule $\overline{jm(C3,M2)}$:}
    \begin{mytikz}  
	\single (c3) at (0,0) {:$C3$}; 
	\single (m2) at (0,-1.5) {:$M2$}; 
	\begin{lhsonly} 
	\edge (c3) edge[bend left=90] node[uselab,text=blue]{:$\overline{in32}$} (m2);
	\end{lhsonly}
    \begin{rhsonly}
    \edge (c3) edge[bend left=0] node[uselab,text=green!70!black]{:$in32$}(m2);    	
    \end{rhsonly}
 	\end{mytikz}
};
\end{mytikz}
\\
\end{tabular}}
\vspace{.1cm}

%% file: 04-Equivalence.tex

\section{Equivalence of enriched and original grammars}
\label{sec:equivalenceOfGrammars}

All along this section let $\mathcal{CG} = \langle TG, G_{in}, P, \pi, \nacs\rangle$ be a grammar and $E(\mathcal{CG}) = \langle \ov{TG}, \ov{G}_{in}, P', \pi', \nacs'\rangle$ be the corresponding enriched grammar of Definition~\ref{def:encodedGrammar}.
We show  that there is a grammar morphism from  $E(\mathcal{CG})$ to   $\mathcal{CG}$. Interestingly, this morphism not only preserves but also reflects derivations and, we conjecture, sequential independence. 

\begin{proposition}
\label{pr:grammarMorphism}
Let $e_{TG}: (\ov{TG} \lrel {in_{TG}}  TG \rrel {id} TG)$ be a span of type graphs, and $e_{P}: P' \to P$ be the mapping that associates each rule obtained from Construction~\ref{def:rules-inv}
with 
the original rule in $P$. Then $e = \langle e_{TG}, e_P \rangle: E(\mathcal{CG}) \to \mathcal{CG}$ is a well-defined morphism. 
\end{proposition}

\begin{proof}[Proof sketch]
Note that the pullback functor $in_{TG}^{<}$ deletes from a graph all the complemented items, which are typed over $\ov{TG}$ but not over $TG$. The proof proceeds by analyzing all the transformations of the previous section showing that only complemented items are ever added (to the left- and right-hand sides of rules, to constraints and to the start graph). Only the transformation for shape $E$ in step 4 of Construction~\ref{def:rules-inv} may add new non-complemented nodes to a rule $p$. Such nodes are preserved by $p$, and become isolated after retyping along $in_{TG}$. Therefore the retyped rule $e_{TG}^{<>}(p)$ is derived from the target rule $e_{P}(p)$, as desired.\\
Furthermore the NACs are reflected by $e$: this follows from the observation that the NACs of a rule $p$ of $E(\mathcal{CG})$ can be  obtained by lifting all and only the NACs of the original rule $e_{P}(p)$ along the morphism embedding the left-hand side $L_{e_{P}(p)}$ in $L_{p}$. 
\end{proof}

The next result essentially shows that the encoding of NACs is correct: the NACs in  
$E(\mathcal{CG})$ are redundant as they can never be violated by a match to a graph satisfying the invariant.

\begin{lemma}[extended NACs are redundant]
\label{pr:redundantNACs}
Let $p \in P'$ be a rule of $E(\mathcal{CG})$,  $G \in \mathbf{Graph}_{\ov{TG}}$ be a graph such that $G \models \Phi_{inv}$, and  $m: L_{p} \to G$ be a match. Then  $m \models \nacs'(p)$.
\end{lemma}

\begin{proof}[Proof sketch] By Construction~\ref{def:compProds}, each constraint $n: L_{p} \to N$ is lifted  along $L_{p} \to \ov{N}$ obtaining an extended constraint $\ov{N} \to \tilde{N}$. Therefore $\tilde{N}$  contains both the body $N^{-}$ of the shape $n^{-}$ and its complement $\ov{N}^{-}$, violating the invariant $\Phi_{inv}$, and this fact cannot be changed by other possible transformations of the constraint in Construction~\ref{def:rules-inv}, because they can only extend it further.   
We can conclude that $G$ cannot violate the constraint because there cannot be an injective morphism from a graph that doesn't satisfy $\Phi_{inv}$ to a graph that satisfies it, and
  $G$ satisfies the invariant by assumption. 
\end{proof}

We now show that, as desired, morphism $e$ also reflect derivations.
\begin{theorem}[reflection of derivations]
\label{pr:reflection}
The conditional grammar morphism $e: E(\mathcal{CG}) \to \mathcal{CG}$ reflects derivations: if  $G \Rrel {p,m} H$ in $\mathcal{CG}$, then there are graphs $G'$ and $H'$ over $\ov{TG}$, a rule $p' \in P'$ and a match $m': L_{p'} \to G'$ such that $G' \Rrel {p',m'} H'$   in   $E(\mathcal{CG})$, $e_{TG}^{<>}(G') \cong G$, $e_{TG}^{<>}(H') \cong H$, and
$e_{p}(p') = p$.
\end{theorem}

\begin{proof}[Proof sketch]  Let $G' = \inv{G}$. Since $G \Rrel {p,m} H$ in $\mathcal{CG}$, it is possible to show that the rule obtained by complementing $\langle p, \nacs(p)\rangle$ according to Construction~\ref{def:compProds}
has a match $m'$ in $G$. In fact, all the constraints of $p$ are satisfied by $m$ by assumption, meaning that the corresponding shape bodies are not present in $\inv{G}$, and thus by the invariant their complement bodies are present in $\inv{G}$. The complemented rule could have been transformed further by Construction~\ref{def:rules-inv}: a case analysis shows that its left-hand side may only be extended, and in such case the match  $m'$ can be extended as well because the additional structure must be present in $\inv{G}$ by the invariant. Finally if $p$ has a corresponding family of enriched rules by step 4 of Construction~\ref{def:rules-inv} then the representative having the largest match in $\inv{G}$ has to be chosen. If $p'$ and $m'$ are the rule and match found via this procedure, then $\inv{G} \Rrel{p',m'} H'$ because (a) the applicability to the $TG$-typed part of $\inv{G}$ is given by assumption; (b) the addition of complemented items in Constructions~\ref{def:compProds} and~\ref{def:rules-inv} does not introduce new constraints, and (c) the NACs of  $E(\mathcal{CG})$  are redundant by Lemma~\ref{pr:redundantNACs}. Finally, $e_{TG}^{<>}(\inv{G}) \cong G$ by the properties of $\inv{\_}$, $e_{p}(p') = p$ by construction, and 
$e_{TG}^{<>}(H') \cong H$ because $e$ is a grammar morphism.
\end{proof}

We conjecture that the two grammars are equivalent  in an even stronger sense, because sequential independence is not only preserved by morphism $e$, but also reflected. This will be a topic of future investigation. 
To complete the picture sketched in the introduction,  we formalize the idea of ``dropping the redundant NACs'' from $E(\mathcal{CG})$ with a grammar morphism, useful to guarantee that the resulting grammar is equivalent in the strong sense described above.

\begin{proposition}[forgetting the NACs preserves and reflects derivations and independence]
\label{pr:reflect-again}
Let $DE(\mathcal{CG})$ be the conditional grammar obtained from $E(\mathcal{CG})$ by deleting all NACs, i.e., such that $\nacs(p) = \emptyset$ for each rule $p$. Let $d: E(\mathcal{CG}) \to DE(\mathcal{CG})$ be defined as $d = \langle(\ov{TG} \lrel {id} \ov{TG} \rrel {id} \ov{TG}),  id_{P'}\rangle$. Then $d$ is a well-defined morphism. Furthermore it reflects derivations and sequential independence. 
\end{proposition}

\begin{proof}[Proof sketch] Conditions  1 and 2 of grammar morphism hold trivially, and condition 3 holds vacuously because $DE(\mathcal{CG})$ has no NACs. Reflection of derivations holds for the double-pushout part trivially, and for the satisfaction of NACs because they are redundant in $E(\mathcal{CG})$. Reflection of sequential independence holds for the existence of certain morphims and commutativity requirements because retyping is along an iso, and for the additional conditions involving NACs again because they cannot be violated by graphs satisfying the invariant. 
\end{proof}

%

%% file: 04-SymbolicAttributed.tex
\section{Encoding of NACs in Unsafe Attributed Graph Grammars}
\label{sec:AttributedEncoding}
Many complex models require a formalism more expressive than safe grammars over typed graphs. In particular, safety limits instance graphs to be subgraphs of the type graph and hence, assuming this is finite, only allows us to model systems with a finite state space. Real-world applications also usually require that structural features, conveniently represented by nodes and edges in a graph, are augmented by data. This combination of structure and data leads us to attributed graph grammars.

In this section we will discuss the problem of encoding a conditional attributed graph grammar into an attributed graph grammar without NACs. The constructions proposed are applicable to all mainstream notions of attributed graph transformation, in particular those based on E-graphs~\cite{DBLP:series/eatcs/EhrigEPT06} and ground symbolic attributed graphs~\cite{LL10}. 

As in the non-attributed case we work with typed graphs over a given type graph declaring, apart from node and edge types, also the attributes for nodes and edges and their respective domains given by the sorts of a given data algebra. The data algebra is fixed across all grammars, representing predefined basic data types such as strings and numbers.

Attributed graph transformation is defined based on the transformation of the underlying graph structure with the possibility of attribute constraints restricting possible matches and attribute assignments determining the update of attribute values. In the following subsections we consider an example, describe the encoding for attributed condition grammars first in general and then apply it to the example.

\subsection{Case Study: Token Curated Registry}

In this section we show how our encoding  applies to an unsafe grammar using a model of a Token-Curated Registry (TCR)  based on \cite{tcr1}.
TCR is an architectural pattern for smart contracts where a set of curators ensure the quality of the entries in a list of services or products.  %
A typical scenario starts with a candidate applying to a registry to be listed, if they are not on the list yet. A curator of the registry can challenge a candidate if they are not being challenged at the time. Every curator of the registry hosting a challenge can vote provided that they are not the challenger and have not already voted. 

For simplicity, we assume there is a single registry, and we have majority voting. That means, for the challenge $ch$ to succeed, the number of voters in support of the challenge must be larger than half of the number of potential voters (all curators of the registry except the challenger), that is  $maj(r, ch) := 1~~ if~~ \left(ch.noVotes > \frac{r.noCurs-1}{2}\right)~~ else~~ 0$. Note that 
by this definition in the event of a tie the challenged candidate is considered the winner.  


Curators who supported the majority position receive a reward. In the start graph $G_0$, all reward attributes are initialised with zero. For all objects, the $rwds$ attribute value must be non-negative at all times. We left this condition implicit here but it can easily be asserted for every rule that subtracts from $rwds$ of an object.


\begin{center}
\scalebox{.7}{ 
\begin{tabular}{c@{\hspace{3cm}} c}
\begin{mytikz}
\graphsingle (g1) {
\graphname{type graph $TG$:}
    \begin{mytikz}  
\usedouble (reg) at (0,0) {$Registry$\n noCurs:Int};
\usedouble (cand) at (-4.5,0) {$Candidate$\n rwds:Int};
 \usedouble (cur) at (0,-2.5) {$Curator$\n rwds:Int};
 \usedouble (ch) at (-4.5,-2.5) {$Challenge$\n noVotes:Int\\noRwds:Int};
 \useedge (cand) edge[bend left=-0]  node[uselab]{$on$}(reg);
  \useedge (cur) edge[bend left=-0]  node[uselab]{$of$}(reg);
 \useedge (cur) edge[bend left=-0]  node[uselab]{$voteYay$} (ch);
  \useedge (ch) edge[bend left=-0]  node[uselab]{$ch'ed$} (cand);
    \useedge (ch) edge[bend left=-20]  node[uselab]{$ch'er$} (cur);
    \useedge (ch) edge[bend left=20] node[uselab]{$rwd$} (cur);
		\end{mytikz}
};
\end{mytikz}
&
 \begin{mytikz}
\graphsingle (sg) {
\graphname{start graph $G_{0}$:}
    \begin{mytikz} 	
	\double (reg) at (0,-1.25) {r:$Registry$\n noCurs = 4};  
	\single (cand) at (0,0) {c:$Candidate$};  
	\single (cur1) at (4,0) {t1:$Curator$};  
	\single (cur2) at (4,-.75) {t2:$Curator$};  
	\single (cur3) at (4,-1.5) {t3:$Curator$};  
	\single (cur4) at (4,-2.25) {t4:$Curator$};  
	\edge (cur1) -- node[uselab]{:$of$}(reg);
	\edge (cur2) -- node[uselab]{:$of$}(reg);
	\edge (cur3) -- node[uselab]{:$of$}(reg);
	\edge (cur4) -- node[uselab]{:$of$}(reg);

 	\end{mytikz}
};
\end{mytikz}\\
\end{tabular}}
\end{center}
\vspace{.1cm}
%

In a challenge, either the challenger (curator) wins or the challenged (candidate). In the first case the candidate is dropped from the list whereas in the second case it stays on. In any case the winner and curators who supported the winning party get a reward. 
In this example all constraints are incremental, thus if in a rule there is more than one red edge, each of them has to be considered as a distinct constraint.  
%

\begin{center}
\scalebox{.75}{
\begin{tabular}{l @{\hspace{1cm}}l@{\hspace{1cm}} l}
\begin{mytikz}
\graphsingle (g1) {
\graphname{rule $apply$:}
    \begin{mytikz}  
	\double (cand) at (0,-2.5) {:$Candidate$};
	\double (reg) at (0,0) {r:$Registry$};   
    \begin{naconly}
    \edge (cand) -- node[uselab, text=red]{:$on$}(reg);
    \end{naconly}
    \begin{rhsonly}
    \edge (cand) edge[bend left=30] node[uselab, text=green!70!black]{:$on$}(reg);
    \end{rhsonly}
 	\end{mytikz}
};
\end{mytikz}
 &
\begin{mytikz}
\graphsingle (g1) {
\graphname{rule $challenge$:}
    \begin{mytikz}  
	\double (cur) at (4,0) {:$Curator$};
	\double (reg) at (0,0) {r:$Registry$};   
	\single (cand) at (0,-3) {:$Candidate$};
	\edge (cand) edge[bend left=0,looseness=1] node[uselab]{:$on$}(reg);
	\edge (cur) -- node[uselab]{:$of$}(reg);
    \begin{naconly}
    \double (ch1) at (2,-1.5){:$Challenge$};
    \edge (ch1) -- node[uselab,text=red]{:$ch'ed$} (cand);
    \end{naconly}
    \begin{rhsonly}
    \double (ch) at (4,-3) {:$Challenge$\n noVoters = 0\\ noRwds = 0};
    \edge (ch) -- node[uselab,text=green!70!black]{:$ch'ed$} (cand);
    \edge (ch) -- node[uselab,text=green!70!black]{:$ch'er$} (cur);
    \end{rhsonly}
 	\end{mytikz}
};
\end{mytikz}
& 
\begin{mytikz}
\graphsingle (g1) {
\graphname{rule $voteYay$:}
    \begin{mytikz}  
	\single (cur) at (0,3) {:$Curator$};
	\single (reg) at (-4,3) {:$Registry$};   
	\single (cand) at (-4,0) {:$Candidate$};
    \double (ch) at (0,0) {:$Challenge$\n $noVotes \to noVotes + 1$};
    \edge (ch) -- node[uselab]{:$ch'ed$} (cand);
	\edge (cand) -- node[uselab]{:$on$}(reg);
	\edge (cur) -- node[uselab]{:$of$}(reg);
    \begin{naconly}
	\edge (ch) -- node[uselab,text=red]{:$ch'er$} (cur);
	\path[->,red] (cur) edge[bend left = -40, looseness=1.5] node[uselab,text=red,left]{:$voteYay$} (ch);
    \end{naconly}
    \begin{rhsonly}
	\edge (cur) edge[bend left = 40, looseness=1.5] node[uselab,text=green!70!black,right]{:$voteYay$} (ch);
    \end{rhsonly}
 	\end{mytikz}
};
\end{mytikz}
\end{tabular}}
\end{center}  
\spv

\begin{center}
\scalebox{.7}{ 
\begin{tabular}{ll}
\begin{mytikz}
\graphsingle (g1) {
\graphname{rule $rewardCh'er$:}
    \begin{mytikz}  
	\double (cur1) at (0,0) {:$Curator$\n $rwds \to rwds + 1$};
	\double (reg) at (4,0) {r:$Registry$};
	\double (cand) at (4,-1.5) {:$Candidate$};
    \double (ch) at (0,-1.5) {ch:$Challenge$};
    \edge (cur1) -- node[uselab]{:$of$}(reg);
   	\edge (ch) -- node[uselab]{:$ch'er$} (cur1);
    \begin{lhsonly}
    \edge (ch) -- node[uselab,text=blue]{:$ch'ed$} (cand);
   	\edge (cand) -- node[uselab,text=blue]{:$on$}(reg);
	\end{lhsonly}
 	\end{mytikz} 	
};
\node(n1)[below=.1cm of g1]{$maj(r, ch) == 1$};
\end{mytikz}
&
\begin{mytikz}
\graphsingle (g1) {
\graphname{rule $rewardCh'ed$:}
    \begin{mytikz}  
	\double (cur1) at (0,0) {:$Curator$};
	\double (reg) at (4,0) {r:$Registry$};
	\double (cand) at (4,-1.5) {:$Candidate$\n $rwds \to rwds + 1$};
    \double (ch) at (0,-1.5) {ch:$Challenge$};
    \edge (cur1) -- node[uselab]{:$of$}(reg);
   	\edge (cand) -- node[uselab]{:$on$}(reg);
   	\edge (ch) -- node[uselab]{:$ch'er$} (cur1);
    \begin{lhsonly}
    \edge (ch) -- node[uselab,text=blue]{:$ch'ed$} (cand);
	\end{lhsonly}
 	\end{mytikz}
};
\node(n1)[below=.1cm of g1]{$maj(r, ch) == 0$};
\end{mytikz}\\
\end{tabular}
}
\end{center}

The $vote$ rule is disabled once either the challenger or the challenged candidate is rewarded by the above rules since the link from the challenged candidate to the challenge is deleted at this stage. By deleting this link, the rules for rewarding the challenger or challenged also enable the rewards for curators who voted in favour or against the proposal.

%
\begin{center}  
\begin{longtable}{c c}
\scalebox{.7}{ 
\begin{mytikz}
\graphsingle (g3) {
\graphname{rule $rewardVoter$:}
    \begin{mytikz}  
	\double (cur1) at (2,-1.5) {t:$Curator$\n $rwds \to rwds + 1$};
	\double (reg) at (2,0) {r:$Registry$};   
    \double (ch) at (-2,-1.5) {ch:$Challenge$};
       \edge (cur1) -- node[uselab]{:$of$}(reg);
       	\double (cand) at (-2,0) {:$Candidate$};
	  \edge (cand) -- node[uselab]{:$on$}(reg);
    \begin{lhsonly}
	\edge (cur1) -- node[uselab,text=blue]{:$voteYay$} (ch);
	\end{lhsonly}
    \begin{naconly}
	 \edge (cand) -- node[uselab,text=red]{:$ch'ed$}(ch);
	 \end{naconly}

 	\end{mytikz}
 	 	};
 	 	\node(n1)[below=.1cm of g3]{$maj(r, ch) == 1$};
\end{mytikz} }
& 
\scalebox{.7}{ 
\begin{mytikz}
\graphsingle (g3) {
\graphname{rule $rewardNonVoter$:}
    \begin{mytikz}  
	\double (cur1) at (2.5,-2.5) {t:$Curator$\n $rwds \to rwds + 1$};
	\double (reg) at (2.5,0) {r:$Registry$};   
    \double (ch) at (-2.5,-2.5) {ch:$Challenge$\n $noRwds \to noRwds + 1$};
       \edge (cur1) -- node[uselab]{:$of$}(reg);
       	\double (cand) at (-2.5,0) {:$Candidate$};
	  \edge (cand) -- node[uselab]{:$on$}(reg);
       \begin{naconly}
        \edge (cand) -- node[uselab,text=red]{:$ch'ed$}(ch);
       \edge (cur1) -- node[uselab,text=red]{:$voteYay$} (ch);
       \path[->] (ch) edge[color=red,bend left=20] node[uselab,text=red]{:$rwd$} (cur1);
       \path[->] (ch) edge[color=red,bend right=20] node[uselab,text=red]{:$ch'er$} (cur1);
       \end{naconly}
       \begin{rhsonly}
       \edge (ch) edge[bend left=40] node[uselab,text=green!70!black]{:$rwd$} (cur1); 
       \end{rhsonly}
 	\end{mytikz}
 	 	};
 	 	\node(n1)[below=.1cm of g3]{$maj(r, ch) == 0$};
\end{mytikz}}\\
\end{longtable}
\begin{longtable}{c c c}
\scalebox{.7}{ 
\begin{mytikz}
\graphsingle (g3) {
\graphname{rule $resolveChallenge$:}
    \begin{mytikz}  
	\double (reg) at (-6,-3) {r:$Registry$};  
	  \double (cur1) at (-3,-1.5) {:$Curator$};
   	 \edge (cur1) -- node[uselab]{:$of$}(reg);
	\begin{lhsonly} 
    \double (ch) at (-3,-3) {ch:$Challenge$};
  	\edge (ch) -- node[uselab,text=blue]{:$ch'er$} (cur1);     
       \end{lhsonly}
 	\end{mytikz}
 	 	};	 	
\end{mytikz} }
&
\scalebox{.7}{ 
\begin{mytikz}
\graphsingle (g3) {
\graphname{rule $delVoteLink$:}
    \begin{mytikz}  
	\double (cur1) at (-3,0) {t:$Curator$};
	\double (reg) at (0,0) {r:$Registry$};   
    \double (ch) at (-2,-1.5) {ch:$Challenge$};
       \edge (cur1) -- node[uselab]{:$of$}(reg);
    \begin{lhsonly}
	\edge (cur1) -- node[uselab,text=blue]{:$voteYay$} (ch);
	\end{lhsonly}
 	\end{mytikz}
 	 	};
 	 	\node(n1)[below=.1cm of g3]{$maj(r, ch) == 0$};
\end{mytikz} }
&
\scalebox{.7}{ 
\begin{mytikz}
\graphsingle (g3) {
\graphname{rule $delRwdLink$:}
    \begin{mytikz}  
	\double (cur1) at (-3,0) {t:$Curator$};
	\double (reg) at (0,0) {r:$Registry$};   
    \double (ch) at (-2,-1.5) {ch:$Challenge$};
       \edge (cur1) -- node[uselab]{:$of$}(reg);
    \begin{lhsonly}
	\edge (ch) -- node[uselab,text=blue]{:$rwd$} (cur1);
	\end{lhsonly}
 	\end{mytikz}
 	 	};
 	 	\node(n1)[below=.1cm of g3]{$maj(r, ch) == 0$};
 	 	\node(n2)[below=.1cm of n1]{$ch.noRwds == r.noCurs - ch.noVotes$};
\end{mytikz} }\\
%
%
%
\end{longtable}
\end{center}  
There are two scenarios after  voting has finished. 
If the challenge succeeds, i.e. $maj(ch, r) == 1$
\begin{itemize}
\item Every voter is rewarded and the corresponding vote link is deleted. 
\item Once all vote links are deleted, the challenge can be resolved by rule $resolveChallenge$. Due to the dangling condition, this rule is only enabled once all vote links are deleted.
\end{itemize}
If the challenge fails, i.e. $maj(ch, r) == 0$
\begin{itemize}
\item Every non-voter who isn't the challenger is rewarded and a reward link is pointed to them from the challenge to make sure none is rewarded more than once. Here, the vote links are intact through the rewarding process. 
\item To safely resolve the challenge, we delete all reward and vote edges by applying $delRwdLink$ and $delVoteLink$ repeatedly. The order does not matter. $delRwdLink$ is only enabled once all non-voters have been rewarded (i.e. $ch.noRwds == r.noCurs - ch.noVotes$). This condition does not change when $delRwdLink$ is applied so the rule remains enabled until all reward edges are deleted.
\item And finally apply $resolveChallenge$ as before.\\
\end{itemize}

\subsection{Encoding of NACs}

Next, we discuss how the NACs of a conditional attributed grammar like the TCR model can be encoded. This encoding extends the analogous construction for safe unattributed grammars. First, let us look at the simpler case of complementation in the case of incremental NACs that have unique occurrences. 
%
A NAC $n: L \to N$ has \emph{unique occurrences} if for every reachable graph $G$ and match $m: L \to G$ there is at most one occurrence $q: N \to G$ of $n$ such that $m \circ n = q$. This can be seen as a form of local safety, permitting an encoding using complement edges similar to the safe case.
That means, for all NACs $n$ with unique occurrences we add complement types to the type graph and derive constraints and complemented start graph as described in Definitions~\ref{def:complementedTG},~\ref{def:invariant} and~\ref{def:invariantClosure}. Similarly, we replace each rule $p$ with NAC $n$ by its derived rule obtained by applying $p$ to the complemented version of $n$ as in Construction~\ref{def:compProds}. Then, if a rule $p$ deletes/creates an occurrence of $n$, we extend $p$ to create/delete a parallel occurrence of $n$'s complement as described in Construction~\ref{def:rules-inv} 1/2. 

The encoding starts to diverge from the safe case where deletion and creation of boundary nodes is concerned. If $n$ is of shape IN or OUT and $p$ creates/deletes a boundary node in $v \in L_n$, we extend $p$ to create/delete the complement structure $\overline {N_- \setminus n_-(L_n)}$ with it. If $n$ is of shape $E$ and $p$ creates/deletes a boundary node in $v \in L_n$, we extend $p$ by an occurrence of the complement of the negative edge $e \in N_-$ attached to $v$, and add the other boundary node $u \in L_n$ as a multiobject such that edge $e$ between $v$ and $u$ is created/deleted along with $v$.

To address the full TCR example, we also have to consider cases where the NACs do not have unique occurrences only. In this case, we introduce reference counters to keep track of how many occurrences there are. That means, we compute the body of the NAC as before, but then, if $n$ is of shape $E$:
\begin{itemize}
\item[(1)] We add to $TG$ a complement type $t_{N^-}(e)^{n^-}$ for the type $t_{N^-}(e)$ of the edge $e \in E_{N^-}$ and introduce an attribute $\#n^-: nat$ to the new edge type. We add a constraint $\#n^- = \mbox{card}(n^-)$ as an invariant, where $\mbox{card}(n^-)$ is the number of occurrences of $n^-$ parallel to $\bar n^-$.\footnote{Formally, such a constraint is satisfied in a graph if for all occurrences $o: L^- \to G$ of the shared boundary of $n^-$ and $\bar n^-$, the number of compatible occurrences of $n^-$ equals $\#n^-$.} 
\item[(2)] If there exists a rule that creates or deletes a boundary node in $v \in L^-$, we add an attribute $\#\bar n^-$ to the type of $v$, subject to the invariant $\#\bar n^- = \mbox{card} \{\bar e \mid  \mbox{$e$ is the edge attached to $v$ in $N^-$}\}$ (this is counting the complement edges attached to $v$).
\end{itemize}
If $n$ is of shape IN or OUT:
\begin{itemize}
\item[(3)] We introduce an attribute $\#n^-: nat$ to the type of the single boundary node $v$ in $L^-$. 
\item[(4)] We add constraint $\#n^- = \mbox{card}(n^-)$ as invariant ($\mbox{card}(n^-)$ is the number of occurrences of $n^-$).
\end{itemize}
Then, for each NAC $n$ and rule $p$:
\begin{itemize}
\item[(5)] If $n \in \nacs(p)$ we replace $p$ by the derived rule resulting from $p$'s application to $\bar n$ with $\#n^- := 0$.
\item[(6)] If $p$ creates/deletes $k$ occurrences of the NAC, we extend $p$ to increase/decrease $\#n^-$ by $k$.
\item[(7)] If $n$ is of shape IN or OUT and $p$ creates a boundary node in $v \in L^-$, we extend $p$ by the attribute assignment $\#n^- := 0$.
\item[(8)] If $n$ is of shape $E$ and $p$ creates/deletes a boundary node in $v \in L^-$, extend $p$ by an occurrence of the complement of the negative edge $e \in N^-$ attached to $v$, add the other boundary node $u \in L^n$ as a multiobject such that edge $e$ between $v$ and $u$ is created/deleted along with $v$. Set $\#\bar n^- := 0$ when $v$ is created.  
\end{itemize} 

Finally, extend the start graph to include complements of all missing NAC occurrences, 
add all newly declared attributes and initialise them such that their constraints are satisfied. Together with the adaptations to the rules, this ensures that the constraints are invariants.

Note that we consider rules with multiobjects as rule schemata that unfold into a set of instance rules using amalgamation \cite{amal}, such that for each rule instance with $k$ copies of a multiobject arising from a boundary node of a NAC $n$, the value of the attribute $\# \bar n^-$ on the other boundary node of $n$ is $k$.

\subsection{Encoding the NACs of the TCR}

We apply this construction to the TCR grammar as follows (see the figures below).
As before, the attributes not mentioned in $\overline{G_0}$ are initialised with zero. There are five rules with NACs, namely $apply$, $challenge$, $voteYay$, $rewardVoter$, and $rewardNonVoter$. With the exception of $challenge$ which has a NAC of shape $IN$, all NACs are of shape $E$ for which we add an attributed complement edge in $\overline{TG}$ (step (1) in the encoding). Since the boundary node $Challenge$ is created by $\overline{challenge}$ and deleted by $\overline{resolveChallenge}$, following step (2), we add attributes to $Challenge$ with invariants requiring that each attribute always reflects the number of corresponding edges (e.g. $ch.noVoteBar == card\{voteL:\overline{vote} ~~|~~ tar(voteL) == ch\}$). For the NAC in $challenge$, we follow step (3) and add attribute $noCh$ to the boundary node $Candidate$ that counts the number of arcs of type $ch'ed$ that target it. Based on  (4), we assert the invariant $c.noCh == card\{chL:ch'ed ~~|~~ tar(chL) == c\}$ for every Candidate $c$.    Conditions on attributes are below each rule.

\medskip
\noindent
\begin{minipage}{0.5\textwidth}
\begin{center}
\scalebox{.7}{ 
\begin{mytikz}
\graphsingle (g1) {
\graphname{type graph $\overline{TG}$:}
    \begin{mytikz}  
\usedouble (reg) at (0,0) {$Registry$\n noCurs:Int};
\usedouble (cand) at (-6.5,0) {$Candidate$\n rwds:Int\\noCh:Nat};
 \usedouble (cur) at (0,-3.5) {$Curator$\n rwds:Int};
 \usedouble (ch) at (-6.5,-3.5) {$Challenge$\n noVotes:Int\\noRwds:Int\\ noRwdBar:Nat\\noVoteBar:Nat\\noCh'erBar:Nat};
 \useedge (cand) edge[bend left=10]  node[uselab]{$on$}(reg);
  \useedge (cand) edge[bend left=-10]  node[uselab]{$\overline{on}~[n:Nat]$}(reg);
  \useedge (cur) edge[bend left=-0]  node[uselab]{$of$}(reg);
 \useedge (cur) edge[bend left=-0]  node[uselab]{$voteYay$} (ch);
 \useedge (cur) edge[bend left=15]  node[uselab]{$\overline{voteYay}~[n:Nat]$} (ch);
  \useedge (ch) edge[bend left=-0]  node[uselab]{$ch'ed$} (cand);
  \useedge (ch) edge[bend left=-70]  node[uselab]{$\overline{ch'ed}~[n:Nat]$} (cand);
    \useedge (ch) edge[bend left=-30]  node[uselab]{$ch'er$} (cur);
    \useedge (ch) edge[bend left=15] node[uselab]{$rwd$} (cur);
    \useedge (ch) edge[bend left=30] node[uselab]{$\overline{rwd}~[n:Nat]$} (cur);
      \useedge (ch) edge[bend left=-45,looseness=1]  node[uselab,below]{$\overline{ch'er}~[n:Nat]$} (cur);
		\end{mytikz}
};
\end{mytikz}
}\\[.5cm]

\scalebox{.65}{
\begin{mytikz}
\graphsingle (g1) {
\graphname{rule $\overline{apply}$:}
    \begin{mytikz}  
	\double (cand) at (0,-2) {c:$Candidate$}; 
	\double (reg) at (0,0) {r:$Registry$};
    \begin{rhsonly}
    \edge (cand) edge[bend right=20] node[uselab, text=green!70!black]{:$on$}(reg);
    \end{rhsonly}
    \edge (cand.west) edge[bend left=35] node[uselab]{:$\overline{on}[n \to n+1]$}(reg.west);
 	\end{mytikz}
};
\node(n1) [below=.3cm of g1] {$c.\overline{on}.n==0$};
\end{mytikz}}
 \quad
\scalebox{.65}{
\begin{mytikz}
\graphsingle (g1) {
\graphname{rule $\overline{voteYay}$:}
    \begin{mytikz}  
	\single (cur) at (-1.5,1) {t:$Curator$};
	\double (reg) at (-5,1) {r:$Registry$};   
    \double (ch) at (-1,-1) {ch:$Challenge$\n $noVotes \to noVotes + 1$};
	\single (cand) at (-5,-1) {:$Candidate$};
	\edge (cur) -- node[uselab]{:$of$}(reg);
    \edge (ch) -- node[uselab]{$ch'ed$} (cand);
	\edge (cand) -- node[uselab]{$on$}(reg);
    \begin{rhsonly}
	\edge (cur) edge[bend left=-80] node[uselab,text=green!70!black, left]{:$voteYay$} (ch);
    \end{rhsonly}
	\edge (cur) edge[bend left=50] node[uselab,text=black]{voteL:$\overline{voteYay}~[n \to n+1]$} (ch);
  	\end{mytikz}
};
\node(n1) [below=.3cm of g1] {$voteL.n == 0$};
\end{mytikz}}
\end{center}
\end{minipage}
\begin{minipage}{0.5\textwidth}
\begin{center}
\scalebox{.7}{
\begin{mytikz}
\graphsingle (sg) {
\graphname{start graph $\overline{G_{0}}$:}
    \begin{mytikz} 	
	\double (reg) at (-1,-2) {r:$Registry$\n $noCurs = 4$};  
	
	\double (cand) at (-1,0) {c:$Candidate$\n $noCh = 0$};  
	\single (cur1) at (4,0) {t1:$Curator$};  
	\single (cur2) at (4,-1) {t2:$Curator$};  
	\single (cur3) at (4,-2) {t3:$Curator$};  
	\single (cur4) at (4,-3) {t4:$Curator$};  
	\edge (cur1) -- node[uselab]{:$of$}(reg);
	\edge (cur2) -- node[uselab]{:$of$}(reg);
	\edge (cur3) -- node[uselab]{:$of$}(reg);
	\edge (cur4) -- node[uselab]{:$of$}(reg);
	\edge (cand) edge[bend left=0]  node[uselab]{$\overline{on}~[n=0]$}(reg);
 	\end{mytikz}
};
\end{mytikz}
}\\[.5cm]

\scalebox{.65}{
\begin{mytikz}
\graphsingle (g1) {
\graphname{rule $\overline{challenge}$:}
    \begin{mytikz}  
	\single (cur) at (3,1) {t:$Curator$};
	\double (reg) at (0,1.5) {r:$Registry$};  
	\double (cand) at (0,-2) {c:$Candidate$\n $noCh \to noCh + 1$};
	  \single (cur2n) at (5.4,1.4) {:$Curator$};
	\single (cur2) at (5.5,1.5) [fill=white]{:$Curator$};

	\edge (cand) -- node[uselab]{:$on$}(reg);
	\edge (cur) -- node[uselab]{:$of$}(reg);
	\edge (cur2) edge[bend left=-0] node[uselab]{:$of$}(reg);

    \begin{rhsonly}
    \double (ch) at (5.5,-2) {ch:$Challenge$\n $noVoters = 0$\\ $noRwds = 0$\\$noRwdBar = r.noCurs - 1$\\$noVoteBar = r.noCurs - 1$\\$noCh'erBar = r.noCurs - 1$};
    \edge (ch) -- node[uselab,text=green!70!black]{:$ch'ed$} (cand);
    \edge (ch) edge[bend left=20] node[uselab,text=green!70!black,left]{:$\overline{ch'ed}~[n = 1]$} (cand);
    \edge (ch) -- node[uselab,text=green!70!black]{:$ch'er$} (cur);
    \edge (ch) edge[bend left=20] node[uselab,text=green!70!black,left]{:$\overline{ch'er}~[n = 1]$} (cur);

	\edge (cur2) edge[bend left=0] node[uselab,text=green!70!black]{:$\overline{voteYay}~[n = 0]$} (ch);
	\edge (ch) edge[bend left=-50] node[uselab,text=green!70!black,above]{:$\overline{rwd}~[n = 0]$} (cur2);
 	\edge (ch) edge[bend left=-80] node[uselab,text=green!70!black,right]{:$\overline{ch'er}~[n = 0]$} (cur2);
    \end{rhsonly}
 	\end{mytikz}
};
\node(n1) [below=.1cm of g1] {$c.noCh == 0$};
\end{mytikz}
}
\end{center}
\end{minipage}

\medskip
In $\overline{apply}$, if there is no edge of type $on$ present (i.e. attribute $n$ of $\overline{on}$ is zero), such an edge is created and the attribute of $\overline{on}$ is incremented by one to reflect this change (cf. (1) and (2) of the encoding). In rule $challenge$, if the candidate $c$ is not currently challenged (i.e. $c.noCh == 0$), a challenge is created and $c.noCh$ is incremented by one. In addition to the creation of the $ch'ed$ edge, its complement $\overline{ch'ed}$ is created with counter initialised to $n=1$. Since Challenge is a boundary node for complement edges corresponding to NACs in $voteYay$ and $rewardNonVoter$, we have a multiobject for curators in rule $challenge$ (creates a Challenge node) and $resolveChallenge$ (deletes a Challenge node) (step (8)). 
%
Rules that contain multiobjects can be interepreted as interaction schemes \cite{amal} which expand into countably infinite set of rules with reference count attributes ensuring that the rule with the correct number of instances of the multiobjects is applied in each case.

In the rewarding phase, the rules are the following:\\  

\begin{center}
\scalebox{.7}{ 
\begin{tabular}{ll}
\begin{mytikz}
\graphsingle (g1) {
\graphname{rule $\overline{rewardCh'er}$:}
    \begin{mytikz}  
	\double (cur1) at (0,0) {:$Curator$\n $rwds \to rwds + 1$};
	\double (reg) at (6,0) {r:$Registry$};
	\double (cand) at (6,-2.5) {:$Candidate$\n $noCh \to noCh - 1$};
    \double (ch) at (0,-2.5) {ch:$Challenge$};
    \edge (cur1) -- node[uselab]{:$of$}(reg);
   	\edge (ch) -- node[uselab]{:$ch'er$} (cur1);
	     \edge (ch) edge[bend left=20] node[uselab]{:$\overline{ch'ed}~[n \rightarrow n- 1]$} (cand);
    \begin{lhsonly}
    \edge (ch) -- node[uselab,text=blue]{:$ch'ed$} (cand);
   	\edge (cand) edge[bend right=20] node[uselab,text=blue]{:$on$}(reg);
	\end{lhsonly}

     \edge (cand) edge[bend right=-20] node[uselab,text=black,left] {$\overline{on}~[n \to n - 1]$} (reg);
 	\end{mytikz} 	
};
\node(n1)[below=.1cm of g1]{$maj(r, ch) == 1$};
\end{mytikz}
&
\begin{mytikz}
\graphsingle (g1) {
\graphname{rule $\overline{rewardCh'ed}$:}
    \begin{mytikz}  
	\double (cur1) at (0,0) {:$Curator$};
	\double (reg) at (6,0) {r:$Registry$};
	\double (cand) at (6,-2.5) {:$Candidate$\n $noCh \to noCh - 1$\\$rwds \to rwds + 1$};
    \double (ch) at (0,-2.5) {ch:$Challenge$};
    \edge (cur1) -- node[uselab]{:$of$}(reg);
   	\edge (cand) -- node[uselab]{:$on$}(reg);
   	\edge (ch) -- node[uselab]{:$ch'er$} (cur1);
	     \edge (ch) edge[bend left=20] node[uselab]{:$\overline{ch'ed}~[n \rightarrow n- 1]$} (cand);
    \begin{lhsonly}
    \edge (ch) -- node[uselab,text=blue]{:$ch'ed$} (cand);
	\end{lhsonly}
 	\end{mytikz}
};
\node(n1)[below=.1cm of g1]{$maj(r, ch) == 0$};
\end{mytikz}\\
\end{tabular}
}
\end{center} 

In $\overline{rewardCh'ed}$, we decrement by one the Candidate's reference counter for challenges. In $\overline{rewardCh'er}$, in addition to decrementing $noCh$, we need to decrement the attribute of $\overline{on}$ to reflect deletion of $on$ parallel to it (point 2). Both rules decrement the $\overline{ch'ed}$ counter to reflect the deletion of a $ch'ed$ edge.

\begin{center}
\begin{tabular}{l l}
\scalebox{.75}{ 
\begin{mytikz}
\graphsingle (g3) {
\graphname{rule $\overline{rewardVoter}$:}
    \begin{mytikz}  
	\double (cur1) at (2,-1.5) {t:$Curator$\n $rwds \to rwds + 1$};
	\double (reg) at (2,0) {r:$Registry$};   
    \double (ch) at (-4,-1.5) {ch:$Challenge$};
       \edge (cur1) -- node[uselab]{:$of$}(reg);
         \double (cand) at (-4,0) {:$Candidate$};
	  \edge (cand) -- node[uselab]{:$on$}(reg);
	  	     \edge (ch) edge[bend left=20] node[uselab]{chd:$\overline{ch'ed}$} (cand);
    \begin{lhsonly}
	\edge (cur1) -- node[uselab,text=blue]{:$voteYay$} (ch);
	\end{lhsonly}
	\edge (cur1) edge[bend right=20] node[uselab,text=black,above]{:$\overline{voteYay}~[n \to n - 1]$} (ch);
 	\end{mytikz}
 	 	};
 	 	\node(n1)[below=.1cm of g3]{$maj(r, ch) == 1$};
		\node(n2) [below=.1cm of n1] {$chd.n == 0$};
\end{mytikz} }
&
\scalebox{.65}{ 
\begin{mytikz}
\graphsingle (g3) {
\graphname{rule $\overline{rewardNonVoter}$:}
    \begin{mytikz}  
	\double (cur1) at (4,-2.5) {t:$Curator$\n $rwds \to rwds + 1$};
	\double (reg) at (4,0) {r:$Registry$};   
    \double (ch) at (-2,-2.5) {ch:$Challenge$\n $noRwds \to noRwds + 1$};
       \edge (cur1) -- node[uselab]{:$of$}(reg);
	\edge (ch) edge[bend left=30] node[uselab,text=black]{rwdL:$\overline{rwd}~[n \to n+1]$} (cur1);
	\edge (cur1) edge[bend left=25] node[uselab,text=black]{voteL:$\overline{voteYay}$} (ch);
	\edge (ch) edge[bend left=15] node[uselab,text=black]{chL:$\overline{ch'er}$} (cur1);
	\double (cand) at (-2.5,0) {:$Candidate$};
	  \edge (cand) -- node[uselab]{:$on$}(reg);
	  \edge (ch) edge[bend left=20] node[uselab]{chd:$\overline{ch'ed}$} (cand);
    \begin{rhsonly}
	\edge (ch) edge[bend left=0] node[uselab,text=green!70!black]{:${rwd}$} (cur1);    
    \end{rhsonly}
 	\end{mytikz}
 	 	};
 	 	\node(n1)[below=.1cm of g3]{$maj(r, ch) == 0$};
 	 	\node(n2) [below=.1cm of n1] {$voteL.n == 0$};
 	 	\node(n3) [below=.1cm of n2] {$rwdL.n == 0$};
 	 	\node(n4) [below=.1cm of n3] {$chL.n == 0$};
		\node(n5) [below=.1cm of n4] {$chd.n == 0$};
\end{mytikz} }\\
\scalebox{.75}{ 
\begin{mytikz}
\graphsingle (g3) {
\graphname{rule $\overline{delVoteLink}$:}
    \begin{mytikz}  
	\double (cur1) at (-3,0) {t:$Curator$};
	\double (reg) at (0,0) {r:$Registry$};   
    \double (ch) at (-3,-1.5) {ch:$Challenge$};
       \edge (cur1) -- node[uselab]{:$of$}(reg);
    \begin{lhsonly}
	\edge (cur1) -- node[uselab,text=blue]{:$voteYay$} (ch);
	\end{lhsonly}
    \edge[-] (cur1.west) edge[->,bend right=50] node[uselab,text=black,left]{:$\overline{voteYay}~[n \to n - 1]$} (ch);
 	\end{mytikz}
 	 	};
 	 	\node(n1)[below=.1cm of g3]{$maj(r, ch) == 0$};
\end{mytikz} }
&
\scalebox{.7}{ 
\begin{mytikz}
\graphsingle (g3) {
\graphname{rule $\overline{delRwdLink}$:}
    \begin{mytikz}  
	\double (cur1) at (-3,0) {t:$Curator$};
	\double (reg) at (0,0) {r:$Registry$};   
    \double (ch) at (-2,-2) {ch:$Challenge$};
       \edge (cur1) -- node[uselab]{:$of$}(reg);
    \begin{lhsonly}
	\edge (ch) -- node[uselab,text=blue]{:$rwd$} (cur1);
	\end{lhsonly}
	\edge[-] (ch.west) edge[->,bend left=50] node[uselab,text=black,left]{:$\overline{rwd}~[n \to n - 1]$} (cur1);
 	\end{mytikz}
 	 	};
 	 	\node(n1)[below=.1cm of g3]{$maj(r, ch) == 0$};
 	 	\node(n2)[below=.1cm of n1]{$ch.noRwds == r.noCurs - ch.noVotes$};
\end{mytikz} }\\
\multicolumn{2}{l}{
\scalebox{.75}{ 
\begin{mytikz}
\graphsingle (g3) {
\graphname{rule $\overline{resolveChallenge}$:}
    \begin{mytikz}  
	\double (reg) at (-3,-3) {r:$Registry$};  
	  \double (cur1) at (-3,-1.5) {:$Curator$};
	    \double (cur2n) at (-3.1,-4.6) {:$Curator$};
	  \double (cur2) at (-3,-4.5)[fill=white] {:$Curator$};  
   	 \edge (cur1) -- node[uselab]{:$of$}(reg);
   	  \edge (cur2) -- node[uselab]{:$of$}(reg);
	          \double (cand) at (3,-1.5) {:$Candidate$};
	\begin{lhsonly} 
    \double (ch) at (3,-3) {ch:$Challenge$};
  	\edge (ch) -- node[uselab,text=blue]{:$ch'er$} (cur1);   
  	\edge (ch) edge[bend right=20] node[uselab,text=blue]{chrL:$\overline{ch'er}$} (cur1);   
 	\edge (ch) edge[bend left=-20] node[uselab,text=blue,left]{rwdL:$\overline{rwd}$} (cur2);
  	\edge (ch) edge[bend left=20,looseness=1] node[uselab,text=blue,right]{nchrL:$\overline{ch'er}$} (cur2);
 	\edge (cur2) edge[bend left=-0] node[uselab,text=blue,right]{voteL:$\overline{voteYay}$} (ch);
	 \edge (ch) edge[bend left=-20] node[uselab,text=blue]{chd:$\overline{ch'ed}$} (cand);
       \end{lhsonly}
 	\end{mytikz}
 	 	};	 	
 	 	\node (n1) at (7,1.5) {$rwdL.n == 0$};
 	 	\node (n2) [below=.1cm of n1] {$chrL.n == 1$};
 	 	\node (n3) [below=.1cm of n2] {$nchrL.n == 0$};
 	 	\node (n4) [below=.1cm of n3] {$voteL.n == 0$};
	 	\node (n5) [below=.1cm of n4] {$chd.n == 0$};
 	 	\node (n6) [below=.1cm of n5] {$ch.noCh'erBar == r.noCurs - 1$};
 	 	\node (n7) [below=.1cm of n6] {$ch.noRwdBar == r.noCurs - 1$};
 	 	\node (n8) [below=.1cm of n7] {$ch.noVoteBar == r.noCurs - 1$}; 	
\end{mytikz} }}
\\
\end{tabular}
\end{center}

In $rewardNonVoter$, conditions $voteL.n == 0$, $rwdL.n == 0$, $chL.n == 0$, and $chd.n == 0$ ensure that the curator has not voted, has not been rewarded,  is not the challenger and that the candidate is no longer linked to the challenge. We increment the attribute of $rwdL$ by one to reflect the creation of an edge of type $rwd$. The encoding of $\overline{resolveChallenge}$ gives rise to a multiobject because $Challenge$ is a boundary node for a constraint of shape $E$ in $voteYay$ and $rewardNonVoter$ rules (point 4). The challenger has an edge of type $ch'er$ and a complement edge of type $\overline{ch'er}$ whose attribute is $1$ whereas every non-challenger curator has three complement edges of types $\overline{ch'er}$, $\overline{voteYay}$, and $\overline{rwd}$ each with an attribute value of $0$. The candidate has a $\overline{ch'ed}$ complement edge with counter attribute $0$.

%% file: 05-Conclusion.tex

\section{Conclusion and Future Work}
\label{sec:Conclusion}

As a contribution towards a comprehensive unfolding semantics applicable also to  graph grammars enriched with both attributes and NAC's, in this paper we addressed the problem of encoding the NACs  into the graphical structure of the states. We presented this construction formally for the restricted case of unattributed safe grammars and incremental NACs, for which we proved that the construction generates an unconditional grammar having equivalent derivations of the original one and, we conjecture, manifesting the same sequential independence among transformations. This construction is reminiscent of the complementation of Elementary Net Systems~\cite{DBLP:conf/ac/Rozenberg86}, a construction known to preserve the branching behavior, thus supporting our confidence  that the unfolding semantics of the encoded grammar will correspond precisely to that of the original one. A formal comparison with complementation of nets is left as future work, together with the proof of preservation and reflection of independence. 

The key ideas of how to  generalize the construction to attributed, possibly unsafe grammars was presented informally. Unsafety requires to count the number of items that could cause the violation of a NAC, and thus it requires the encoded grammar to be attributed. More interestingly, certain shapes of NACs require to encode a rule of the original grammar with a family of rules, that can be generated via an amalgamation. As future work we plan to provide a formalization of the construction in this more general case, and to study to what extent the concurrent semantics is preserved through the construction. 

We would also like to generalize our construction to the case of non-incremental NACs, motivated by the observation that in the attributed case incremental NACs have limited expressiveness. But previous work~\cite{incrementalnacs} showed that the notion of independence among transitions is not well defined in case of general NACs, thus weakening our overall motivations in this case. 
